\documentclass[11pt]{article}
\usepackage[utf8]{inputenc}
\usepackage[a4paper,margin=1in]{geometry}
\usepackage{microtype} 
\usepackage{hyperref}
\hypersetup{colorlinks=true,linkcolor=[rgb]{0.75,0,0},citecolor=[rgb]{0,0,0.75}}
\usepackage{adjustbox}
\usepackage{graphicx}
\usepackage{amsmath}
\usepackage{mathrsfs}
\usepackage{amssymb}
\usepackage{color}
\usepackage{threeparttable}
\usepackage{amsthm}
\usepackage{thmtools}
\graphicspath{{./graphics/}}
\declaretheorem[name=Problem]{problem}
\declaretheorem[name=Theorem,numberwithin=section]{theorem}
\declaretheorem[name=Lemma,sibling=theorem]{lemma}
\newcommand{\graphicsScale}{1.}
\newcommand{\butterflyScale}{0.8} 
\newcommand{\grey}[1]{{\textcolor[rgb]{.4,.4,.4}{#1}}}
\newcommand{\green}[1]{{\textcolor[rgb]{0,.506,.0}{#1}}}
\newcommand{\R}{R} 
\newcommand{\B}{B} 
\newcommand{\lt}{\ell} 
\newcommand{\M}{M} 
 
\newcommand{\dist}{d} 
\newcommand{\dd}{\mathbf{d}} 
\newcommand{\D}{\mathbf{D}} 
\newcommand{\PP}{\mathcal{P}} 
\newcommand{\OO}{O} 
\newcommand{\oo}{o} 
\newcommand{\V}{\mathcal{V}}  
\newcommand{\K}{\mathcal{S}} 
\newcommand{\length}[1]{\left|#1\right|} 
\newcommand{\opnint}[2]{\left]#1, #2\right[} 
\newcommand{\sgt}[2]{#1 #2} 
\newcommand{\ray}[2]{#1 #2} 
\newcommand{\lineT}[2]{#1 #2} 
 
\newcommand{\X}{\mathbf{X}} 
\newcommand{\HH}{\mathbf{H}} 
\newcommand{\T}{\mathbf{T}} 
\newcommand{\gS}{\mathbf{S}}

\newcommand{\nbf}{f} 
\newcommand{\nbg}{g} 
\newcommand{\nbh}{\bar h} 
\newcommand{\lft}{\swarrow} 
\newcommand{\ctr}{\downarrow} 
\newcommand{\rgt}{\searrow} 
\newcommand{\pair}[2]{#1,#2} 
\newcommand{\ipair}[2]{\langle#1,#2\rangle} 
\newcommand{\card}[1]{\left|#1\right|} 
\newcommand{\trgl}[3]{#1 #2 #3} 
\newcommand{\Figure}{Figure}
\newcommand{\Figures}{Figures}

\title{Complexity Results on Untangling\texorpdfstring{\\}{} Red-Blue Matchings\texorpdfstring{\thanks{Preliminary versions of this paper appeared on the 38th European Workshop on Computational Geometry (EuroCG'22) and the 15th Latin American Theoretical Informatics Symposium (LATIN'22). This work is partially supported by the IFCAM project Applications of Graph Homomorphisms (MA/IFCAM/18/39), and by the French ANR PRC grant ADDS (ANR-19-CE48-0005).}}{}}
\author{
  Arun Kumar Das\thanks{Indian Statistical Institute, Kolkata, India. arund426@gmail.com, sandipdas@isical.ac.in} 
  \and
  Sandip Das\textsuperscript{\dag} 
  \and 
  Guilherme D. da Fonseca\thanks{Aix-Marseille Université and LIS, France. guilherme.fonseca@lis-lab.fr} 
  \and 
  Yan Gerard\thanks{Université Clermont Auvergne and LIMOS, France. \{yan.gerard, bastien.rivier\}@uca.fr} 
  \and 
  Bastien Rivier\textsuperscript{\S} 
  }
\date{}

\begin{document}
\maketitle

\begin{abstract}
  Given a matching between $n$ red points and $n$ blue points by line segments in the plane, we consider the problem of obtaining a crossing-free matching through flip operations that replace two crossing segments by two non-crossing ones.
  We first show that (i) it is NP-hard to $\alpha$-approximate the shortest flip sequence, for any constant $\alpha$. 
  Second, we show that when the red points are collinear, (ii) given a matching, a flip sequence of length at most $\binom{n}{2}$ always exists, and (iii) the number of flips in any sequence never exceeds $\binom{n}{2}\frac{n+4}{6}$. 
  Finally, we present (iv) a lower bounding flip sequence with roughly $1.5 \binom{n}{2}$ flips, which shows that the $\binom{n}{2}$ flips attained in the convex case are not the maximum, and (v) a convex matching from which any flip sequence has roughly $1.5\, n$ flips. 
  The last four results, based on novel analyses, improve the constants of state-of-the-art bounds. 
\end{abstract}

\section{Introduction}
\label{sec:intro}

We consider the problem of untangling a perfect red-blue matching drawn in the plane with straight line segments.
We are given a set of $2n$ points in the plane, partitioned into a set $\R$ of $n$ red points, and a set $\B$ of $n$ blue points, in general position (no three collinear points, unless they have the same color). 

In combinatorial reconfiguration,
a flip is an  operation changing a configuration into another~\cite{BOSE200960,phdJof}. 
In our case, a \emph{configuration} is a set of $n$ line segments where each point of $\R$ is matched to exactly one point of $\B$, i.e., a perfect straight-line red-blue matching (a \emph{matching} for short), and
a \emph{flip} replaces two crossing segments by two non-crossing ones (\Figure~\ref{fig:flip}). 

\begin{figure}[!ht]
  \centering
  \includegraphics[scale=\graphicsScale]{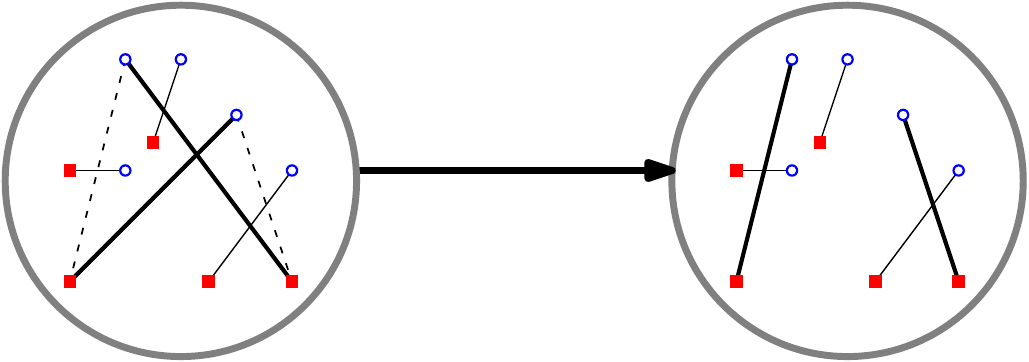}
  \caption{Matchings before and after a flip. Solid squares are red points and hollow circles are blue points.}
  \label{fig:flip}
\end{figure}

The \emph{reconfiguration graph} is the directed simple graph whose vertices $\V$ are the configurations, and such that there is a directed edge from a configuration $\M_1$ to another one $\M_2$ whenever a flip transforms $\M_1$ into $\M_2$.
Note that, in our case, since a flip strictly shortens the total length of the segments (triangle inequality in the two triangles of \Figure~\ref{fig:flip}), the reconfiguration graph is acyclic~\cite{BoM16}.
Let $\K \subseteq \V$ be the set of sinks, which corresponds to the crossing-free matchings. 
Given two configurations $u,v \in \V$, let $\PP(u,v)$ be the set of directed paths from $u$ to $v$. 
Given a path $P$, let the \emph{length} of $P$, denoted $\length{P}$, be the number of edges in $P$. 
We are interested in two parameters of this reconfiguration graph: 
\[
  \dd(\R,\B) = \max_{u \in \V} \min_{v \in \K} \min_{P \in \PP(u,v)} \length{P} 
  \quad \text{ and } \quad 
  \D(\R,\B) = \max_{u \in \V} \max_{v \in \K} \max_{P \in \PP(u,v)} \length{P}.
\]

This leads to the definitions of $\dd(n)$ and $\D(n)$ respectively as the maximum of $\dd(\R,\B)$ and $\D(\R,\B)$ over all sets $\R,\B$ with $\card{\R} = \card{\B} = n$.
An \emph{untangle sequence} is a path in the reconfiguration graph ending in $\K$.
Intuitively, $\dd$ corresponds to the minimum length of an untangle sequence in the worst case, while $\D$ corresponds to the longest untangle sequence.

We also consider a more specific version of the problem where the red points are collinear~\cite{BMS19}, say, on the $x$-axis.
As the flips on each half-plane defined by the $x$-axis are independent, we additionally suppose all blue points to lie on the upper half-plane without loss of generality.
The matchings in this case are called \emph{red-on-a-line} matchings. 

\paragraph{Related Work.}
The parameters $\dd$ and $\D$ have been studied in several different contexts with similar definitions of a flip, but considering other configurations. 

In 1981, an $\OO(n^3)$ upper bound on $\D(n)$ was stated in the context of optimizing a TSP tour~\cite{VLe81} (the configurations are polygons). 
This upper bound should be compared to the exponential lower bound on $\D(n)$ when the flips are not restricted to crossing segments, as long as they decrease the Euclidean length of the tour~\cite{ERV14}.
The convex case (i.e., the case where the points are in convex position) has been studied in~\cite{OdW07,WCL09}. 

In the non-bipartite version of the straight-line perfect matching problem, there are two possible pairs of segments to replace a crossing pair. 
This additional choice yields an $n^2/2$ upper bound on $\dd(n)$~\cite{BoM16}. 

It is also possible to relax the flip definition to all operations that replace two segments by two others with the same four endpoints, whether they cross or not, and generalize the configurations to multigraphs with the same degree sequence~\cite{Hak62,Hak63,phdJof}. 
In this context, finding the shortest path from a given configuration to another in the reconfiguration graph is NP-hard, yet $1.5$-approximable~\cite{BeI08,BeI17,EKM13,Wil99}. If we additionally require the configurations to be connected graphs, the same problem is NP-hard and 2.5-approximable~\cite{BJ20}. 

Reconfiguration problems in the context of triangulations are widely studied~\cite{NiN18}. 
A flip consists of removing one edge and adding another one while preserving a triangulation. It is known that $\Theta(n^2)$ flips are sufficient and sometimes necessary to obtain a Delaunay triangulation~\cite{HNU99,Law72}. 
Determining the flip distance between two triangulations of a point set~\cite{LuP15,Pil14} and between two triangulations of a simple polygon~\cite{AMP15} are both NP-hard.

Considering perfect matchings of an arbitrary graph (instead of the complete bipartite graph on $\R, \B$), a flip amounts to exchanging the edges in an alternating cycle of length four. 
It is then PSPACE-complete to decide whether there exists a path from a configuration to another~\cite{BBH19}. 
There is, actually, a wide variety of reconfiguration contexts derived from NP-complete problems where this same accessibility problem is PSPACE-complete~\cite{TDH11}. 
Many other reconfiguration problems are presented in~\cite{Heu13}.

Getting back to our context of straight-line red-blue matchings, 
the values of $\dd$ and $\D$ have been determined almost exactly in the convex case (see Table~\ref{tab:summary}). 
Notice that the $n-1$ lower bound on $\dd(n)$ carries to both the general and red-on-a-line cases~\cite{BoM16}. 
It is notable that the upper bound on $\D(n)$ is also the best known bound on $\dd(n)$ and has not been improved since 1981~\cite{VLe81}. 

As a final side note, given a red-blue point set, a crossing-free red-blue matching can be computed in $O(n \log n)$ time~\cite{HeSu92}. The algorithm is based on semi-dynamic convex hull data structures and does not use flips. 
The problem has also been considered in higher dimensions~\cite{AkAl89}.

\paragraph{Contributions.}

We show in Section~\ref{sec:NP} that it is NP-hard to $\alpha$-approximate the shortest untangle sequence starting at a given matching, for any fixed $\alpha \geq 1$. 

\begin{table}[tb]
  \centering
  \scalebox{0.93}{
  \begin{threeparttable} 
    \caption{Lower and upper bounds on $\dd(n)$ and $\D(n)$ for red-blue matchings.}
    \label{tab:summary}
      \begin{tabular}{|r|l|l|l|l|}
        \hline
        & \multicolumn{2}{c|}{$\dd(n)$ bounds} & \multicolumn{2}{c|}{$\D(n)$ bounds} \\ \cline{2-5} 
        & \multicolumn{1}{c|}{lower} & \multicolumn{1}{c|}{upper} & \multicolumn{1}{c|}{lower} & \multicolumn{1}{c|}{upper} \\ \hline
        general \rule{0pt}{3ex} & $\frac{3}{2} n - 2$, Thm.~\ref{thm:lowerBd}$^*$& $\binom{n}{2} (n-1)$, \cite{BoM16,VLe81} & $\frac{3}{2}\binom{n}{2} - \frac{n}{4}$, Thm.~\ref{thm:lowerBD}$^*$ & $\binom{n}{2} (n-1)$, \cite{BoM16,VLe81} \\ 
        convex~ & $\frac{3}{2} n - 2$, Thm.~\ref{thm:lowerBd}$^*$ & $2n-2$, \cite{BMS19} & $\binom{n}{2}$, \cite{BoM16} & $\binom{n}{2}$, \cite{BMS19} \\
        red-on-a-line \rule[-1.5ex]{0pt}{0pt} & $n-1$, \cite{BoM16} & $\binom{n}{2}$, Thm.~\ref{thm:algo} & $\frac{3}{2}\binom{n}{2} - \frac{n}{4}$, Thm.~\ref{thm:lowerBD}$^*$ & $\binom{n}{2}\frac{n+4}{6}$, Thm.~\ref{thm:upperB} \\
        \hline
      \end{tabular}
    \begin{tablenotes} 
    \item [$^*$ ] For even $n$.
    \end{tablenotes}
  \end{threeparttable}}
\end{table}

The following results are summarized in Table~\ref{tab:summary}.
An improved lower bound on $\dd(n)$ in the convex case is presented in Section~\ref{sec:fenceLowerB}.
The remainder of the paper considers the red-on-a-line case.
In Section~\ref{sec:algo}, we slightly improve the former $2\binom{n}{2}$ upper bound on $\dd(n)$~\cite{BMS19}, using a simpler algorithm and a novel analysis.
In Section~\ref{sec:upperB}, we asymptotically divide by $6$ the historical $\binom{n}{2} (n-1)$ upper bound on $\D(n)$~\cite{BoM16,VLe81}, using a different potential argument.

In Section~\ref{sec:butterflyLowerB}, we present a counterexample to the intuition that the longest untangle sequence is attained in the convex case (where the number of crossings is maximal). 
We take advantage of points that are not in convex position to increase the lower bound by a factor of $\frac{3}{2}$. 
This red-on-a-line lower bound on $\dd(n)$ carries over to the general case (and even to non-bipartite perfect matchings). 
The conjecture that $\D(n)$ is quadratic~\cite{BoM16} remains open, though.

\section{NP-Hardness}
\label{sec:NP}

In this section, we prove the NP-hardness of the following problem. Let $\dist(\M)$ denote the minimum path length from a matching $\M$ to $\K$, the set of crossing-free matchings, in the reconfiguration graph.

\begin{problem}
  \label{pb:approx} 
  Let $\alpha \geq 1$ be a constant. \\
  {Input:} $\M$, a red-blue matching with rational coordinates. \\
  {Output:} An untangle sequence starting at $\M$ of length at most $\alpha$ times $\dist(\M)$. 
\end{problem}

We have the following theorem.

\begin{theorem}
  \label{thm:NPapx}
  Problem~\ref{pb:approx} is NP-hard for all $\alpha \ge 1$.
\end{theorem}

\paragraph{Reduction Strategy.}

De Berg and Khosravi~\cite{deBerg2012} showed that the rectilinear planar monotone $3$-SAT problem (\emph{RPM $3$-SAT}) is NP-hard. The RPM $3$-SAT problem is a special case of the classic $3$-SAT problem in which the clauses consist only of either all positive or all negative literals and the layout is planar (\Figure~\ref{fig:RPM3CNF}).
We reduce RPM $3$-SAT to Problem~\ref{pb:approx}. 
The key elements of the reduction are described next.

Given a planar embedding of an RPM $3$-CNF formula $\Phi$ (\Figure~\ref{fig:RPM3CNF}), we construct a matching $\M_{\Phi}$ of polynomial size.
The property of this matching $\M_{\Phi}$ is that its shortest untangle sequence has a length below a certain constant if $\Phi$ is satisfiable and above $\alpha$ times this constant otherwise.
Figure~\ref{fig:MPhi} shows the matching $\M_{\Phi}$ corresponding to the formula $\Phi=(x_1 \vee x_2 \vee x_3)\wedge(x_3 \vee x_4 \vee x_5)\wedge(x_3 \vee x_5 \vee x_6)\wedge(\overline{x_2} \vee \overline{x_3} \vee \overline{x_ 4})$ from \Figure~\ref{fig:RPM3CNF}.

The aforesaid matching $\M_{\Phi}$ is built using two types of gadgets. 
The variable rectangles are replaced by variable gadgets (\Figure~\ref{fig:variableGadget}). 
The clause rectangles together with the corresponding edges are replaced with padded clause gadgets. 
A padded clause gadget is represented in \Figure~\ref{fig:clauseGadgetWithPadding} with plain segments. 
Throughout all the figures in Section~\ref{sec:NP}, the dashed segments represent all the possibly created segments after any sequence of flips. 

A \emph{variable gadget} is a three-segment matching with two crossings. It allows for two possible flips, either of which produces a crossing-free matching, as shown in \Figure~\ref{fig:variableGadget}. 
The flip generating the topmost segment stands for \emph{false} ($x=0$ in \Figure~\ref{fig:variableGadget}), while the flip generating the bottom segment stands for \emph{true} ($x=1$). 

A \emph{clause gadget} is an OR gate with three inputs (\Figure~\ref{fig:clauseGadget}). 
The RPM $3$-CNF clauses are either positive or negative. We describe the gadget for a positive clause, but the gadget for a negative clause can be defined analogously (by a vertical reflection).
Three variable gadgets are the inputs of a clause gadget.
In the crossing-free matching obtained for the clause gadget, the presence of the topmost segment ($\sgt{r_4}{b_7}$ in \Figures~\ref{fig:clauseGadgetbis}, \ref{fig:paddingGadget}, \ref{fig:clauseGadgetWithPadding},  and~\ref{fig:clauseGadget}) stands for a false output. 

A \emph{padding gadget} is a gadget that serves to force an arbitrarily large number $k$ of flips if a clause is false. It consists of a series of $k$ non-crossing segments (the plain segments in \Figure~\ref{fig:paddingGadget}, $\sgt{r_4}{b_7}$ aside). 
A \emph{padded clause gadget} is a clause gadget coupled with a padding gadget in such a way that the presence of the output segment triggers $k$ extra flips (\Figure~\ref{fig:clauseGadgetWithPadding}).

Let $c$ be the number of clauses and $v$ be the number of variables of the formula $\Phi$.
If $\Phi$ is \emph{satisfiable}, then the shortest untangle sequence of  $\M_ \Phi$ has at most $5$ flips per clause plus $1$ flip per variable. In this case, we have $\dist(M_ \Phi) \leq 5c + v$.
We choose the size of the padding gadget so that a non-satisfied clause triggers $k = \alpha (5c + v) + 1$ flips. 
If the formula $\Phi$ is \emph{not satisfiable}, then at least one of the padding gadgets is triggered and $\dist(M_ \Phi) > \alpha (5c + v)$.

\paragraph{The Problem to Be Reduced.}

In \emph{RPM $3$-SAT}, the \emph{graph} of a conjunctive normal form (CNF) formula is the bipartite graph with the variables and clauses as vertices, and where there is an edge between a variable and a clause if and only if the clause contains the variable.
A clause is said to be \emph{positive} if it contains only positive variables; it is said to be \emph{negative} if it contains only negative variables.
A CNF formula is \emph{monotone} if each clause is either positive or negative. 

A \emph{rectilinear planar monotone} 3-CNF (\emph{RPM 3-CNF}) formula is a monotone formula with $3$-variables per clause whose graph can be drawn with the following conventions (\Figure~\ref{fig:RPM3CNF}). 
(i) The variables and the clauses are represented by axis-parallel non-overlapping closed rectangles. 
(ii) The variable rectangle centroids lie on the $x$-axis. 
(iii) The positive clause rectangles are above the $x$-axis, the negative ones, below. 
(iv) The edges connecting a variable to a clause are vertical line segments and do not cross any other rectangle. 
We call such a drawing a \emph{planar embedding} of $\Phi$.

\begin{figure}[!ht]
  \centering
  \includegraphics[scale=\graphicsScale]{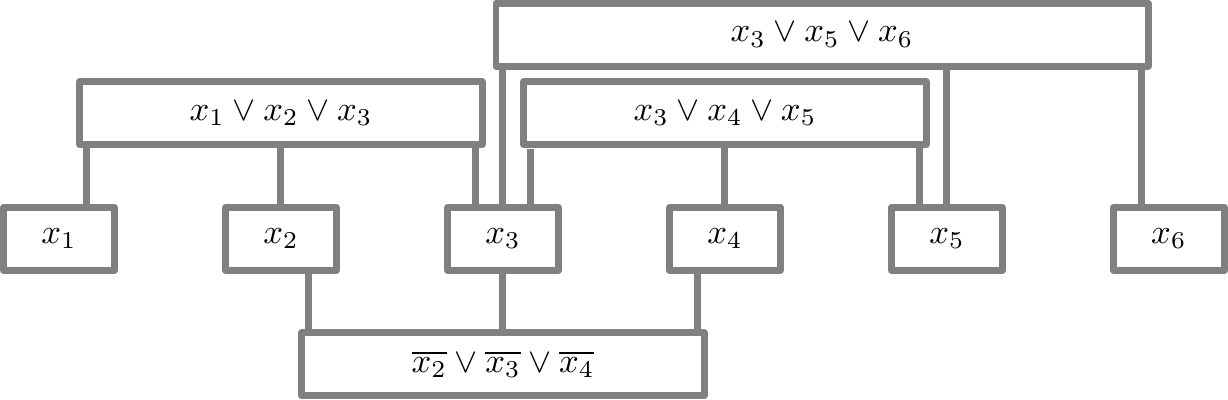}
  \caption{A planar embedding of an RPM $3$-CNF formula $\Phi$.}
  \label{fig:RPM3CNF}
\end{figure}

\begin{figure}[!ht]
  \centering
  \includegraphics[page=2,scale=\graphicsScale]{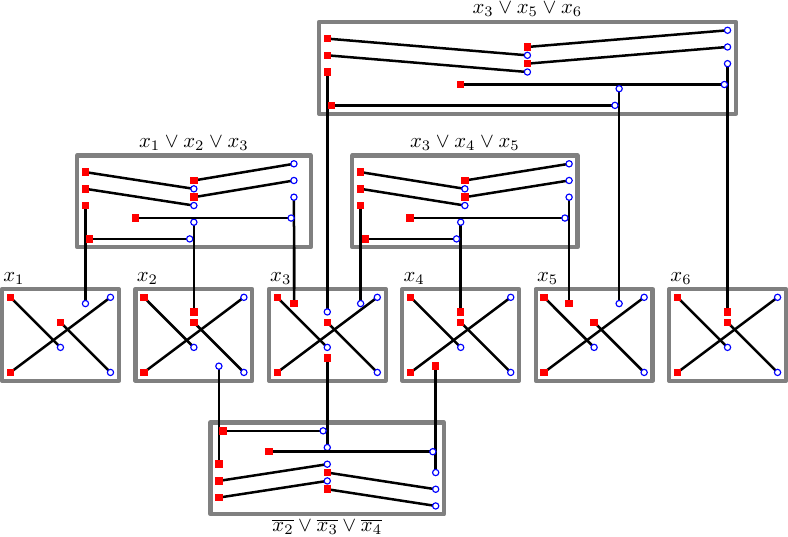}
  \caption{The matching $\M_\Phi$ of the formula $\Phi$ from \Figure~\ref{fig:RPM3CNF}.}
  \label{fig:MPhi}
\end{figure}

\paragraph{Variable Gadgets.} 
A \emph{variable gadget} is a three-segment matching built on the four endpoints of an axis-parallel rectangle as follows (\Figure~\ref{fig:variableGadget}). 
The two leftmost endpoints of the rectangle are colored red, the two rightmost ones are colored blue. 
One of the segments of the matching is the diagonal joining the bottom left red point to the top right blue point. 
We add one red point on the vertical line splitting the rectangle in two symmetric halves, just above the diagonal, in the inside of the rectangle. 
This red point is connected to the bottom right blue point. 
Similarly, we add one blue point on the same vertical, just below the diagonal. 
This blue point is connected to the top left red point.

We will refer to the triangle consisting of the three topmost points of a variable gadget as the \emph{top triangle} of the variable gadget.

\begin{figure}[ht]
  \centering
  \includegraphics[scale=\graphicsScale]{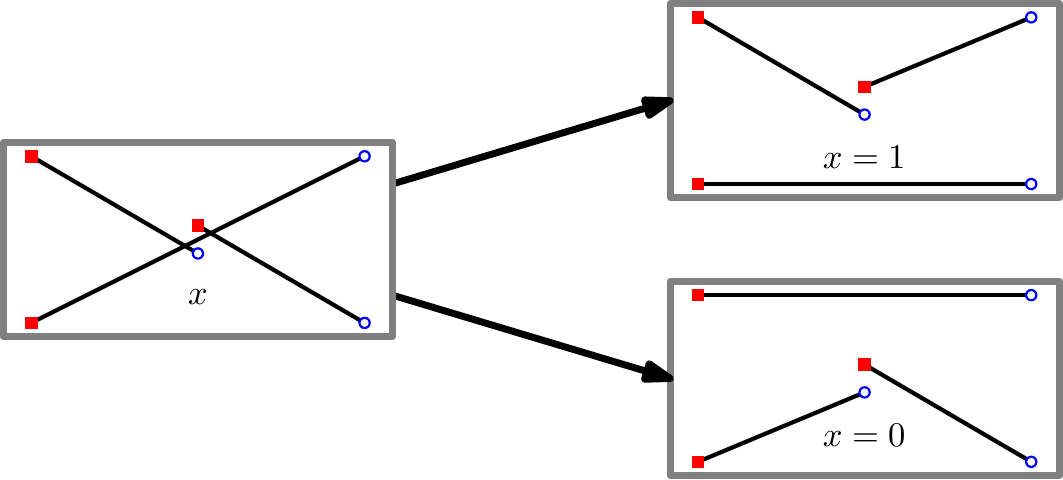}
  \caption{A variable gadget and its two untangle sequences.}
  \label{fig:variableGadget}
\end{figure}

\begin{lemma}
  \label{lem:var}
  A variable gadget is the starting matching of exactly two untangle sequences of length $1$ ending in distinct matchings.
\end{lemma}
\begin{proof}
  It is straightforward to check the two possible cases.
\end{proof}

We can therefore represent each variable $x$ of a propositional formula by a variable gadget.
Assigning $x$ to a truth value amounts to choosing one of the two possible untangle sequences, with the convention that the lower edge of the rectangle is present in the final matching if $x=1$ (i.e., $x$ is ``true''), and that the upper edge of the rectangle is present if $x=0$ (\Figure~\ref{fig:variableGadget}). 

\paragraph{OR Gadgets.}
An \emph{OR gadget} consists of four three-segment matchings built on a common point set, say $\{ r_1, r_2, r'_2, r_3 \}$ for the red points, and $\{ b_1, b'_1, b_2, b_3\}$ for the blue points, as follows (see the first matching in each of \Figures~\ref{fig:ORGadget00}(a), \ref{fig:ORGadget00}(b), \ref{fig:ORGadget00}(c), and \ref{fig:ORGadget00}(d), ignoring the dashed segments).
The \emph{$0\vee 0$ matching} consists of the segments $\sgt{r_1}{b'_1}, \sgt{r'_2}{b_2}, \sgt{r_3}{b_3}$, and only the first two are not crossing.
The \emph{$0\vee 1$ matching} consists of the segments $\sgt{r_1}{b'_1}, \sgt{r_3}{b_3}, \sgt{r_2}{b_2}$, and only the first two are crossing.
The \emph{$1\vee 0$ matching} consists of the segments $\sgt{r'_2}{b_2}, \sgt{r_3}{b_3}, \sgt{r_1}{b_1}$, and only the first two are crossing.
The \emph{$1\vee 1$ matching} consists of the segments $\sgt{r_1}{b_1}, \sgt{r_2}{b_2}, \sgt{r_3}{b_3}$, and is crossing-free.
In addition to these constraints, the point set also satisfies the following ones.
The following three matchings are crossing-free: 
$\{ \sgt{r_1}{b_2}, \sgt{r'_2}{b_3}, \sgt{r_3}{b'_1} \}$, 
$\{ \sgt{r_1}{b_3}, \sgt{r_2}{b_2}, \sgt{r_3}{b'_1} \}$, and
$\{ \sgt{r_1}{b_1}, \sgt{r'_2}{b_3}, \sgt{r_3}{b_2} \}$.
In each of the following two matchings, only the first two segments are crossing: 
$\{ \sgt{r_1}{b_3}, \sgt{r'_2}{b_2}, \sgt{r_3}{b'_1} \}$, and
$\{ \sgt{r_1}{b'_1}, \sgt{r_3}{b_2}, \sgt{r'_2}{b_3} \}$.

Note that, in any of the four matchings of an OR gadget, there is one unused blue point and one unused red point. 
If the unused blue point is $b_1$ (respectively $b'_1$), we say that the \emph{left input} of the OR gadget is $0$ (respectively $1$). 
Similarly, if the unused red point is $r_2$ (respectively $r'_2$), we say that the \emph{right input} of the OR gadget is $0$ (respectively $1$).
To complete the similarity with a logical gate, we also define the \emph{output} of the OR gadget as $0$ if the segment $\sgt{r_1}{b_2}$ is present in all the final matchings of any untangle sequence starting at the OR gadget and as $1$ if the segment $\sgt{r_1}{b_2}$ is absent of all the same final matchings. The output is undefined otherwise.
The following lemma states that the truth table of the logical gate associated with an OR gadget is indeed the one of an OR gate.

We will refer to the smallest of the triangles consisting of the segment $\sgt{r_1}{b_2}$ and induced by all the other segments we have mentioned in the definition of an OR gadget as the \emph{top triangle} of the OR gadget. It is the shaded triangle in \Figure~\ref{fig:ORGadget00}(d).

\begin{figure}[ht]
  \centering
  \includegraphics[scale=\graphicsScale,page=1]{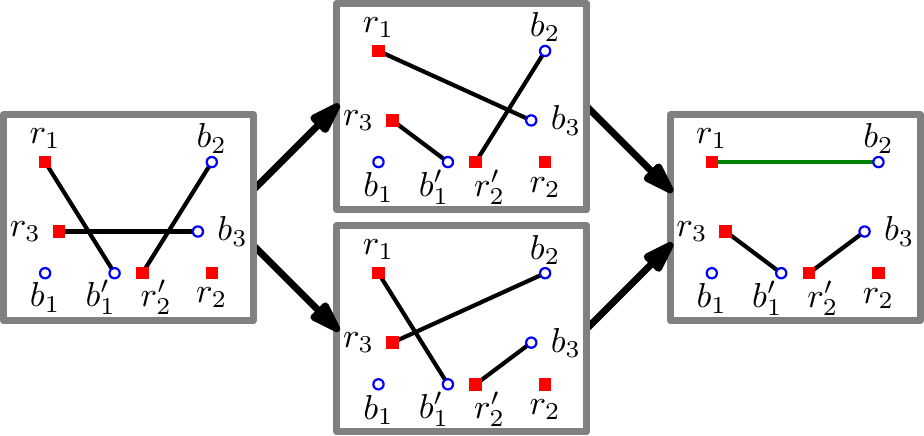}\qquad \qquad
  \includegraphics[scale=\graphicsScale,page=4]{ORGadget}\\
  \hspace{0.2\textwidth}(a) $0\vee0$ \hspace{0.36\textwidth}(d) $1\vee1$ \\\vspace{\baselineskip}
  \includegraphics[scale=\graphicsScale,page=2]{ORGadget}\qquad \qquad
  \includegraphics[scale=\graphicsScale,page=3]{ORGadget}\\
  (b) $0\vee1$ \hspace{0.38\textwidth}(c) $1\vee0$
  \caption{The four matchings of an OR gadget, with their untangle sequences.}
  \label{fig:ORGadget00}
\end{figure}

\begin{lemma} \label{lem:ORGadget}
  The output of an OR gadget is always well defined, and is $0$ if and only if the two inputs of the OR gadget are both $0$. More precisely, we have the following.
  \begin{enumerate} 
  \item The $0\vee 0$ matching is the starting matching of exactly two untangle sequences, each of length $2$, and ending at the same matching containing the upper segment $\sgt{r_1}{b_2}$ (\Figure~\ref{fig:ORGadget00}(a)).
  \item The $0\vee 1$ matching is the starting matching of a unique untangle sequence of length $1$ ending at a matching excluding the upper segment $\sgt{r_1}{b_2}$ (\Figure~\ref{fig:ORGadget00}(b)).
  \item The $1\vee 0$ matching is the starting matching of a unique untangle sequence of length $1$ ending at a matching excluding the upper segment $\sgt{r_1}{b_2}$ (\Figure~\ref{fig:ORGadget00}(c)).
  \item The $1\vee 1$ matching is already crossing free. It excludes the upper segment $\sgt{r_1}{b_2}$ (\Figure~\ref{fig:ORGadget00}(d)).
  \end{enumerate}
\end{lemma}

\begin{proof}
  \label{lem:OR}
  For each of the four $x \vee y$ matchings whith $x,y \in \{0,1\}$, we enumerate all the possible untangling sequences. These sequences are all shown in \Figure~\ref{fig:ORGadget00}. Lemma~\ref{lem:ORGadget} then follows.
\end{proof}

\paragraph{Clause Gadgets.}
A \emph{clause gadget} consists of two OR gadgets, the output of the first one being ``connected'' to the left input of the second one (\Figure~\ref{fig:clauseGadgetbis}). 
More precisely, a clause gadget is built on seven red points, say $r_4, r_5, r_6, r_7, r_8, r_{10}, r_{11}$, and six blue points, say $b_4, b_5, b_6, b_7, b_8, b_9$ such that the following maps correspond to two OR gadgets (using the OR gadget previous notations), and such that $r_8$ lie in the inside of the top triangle of the first OR gadget and is the only overlap between the two OR gadgets.
\begin{align*}
  \text{First OR gadget: } (r_4, b_4, r_5, b_5, r_6, b_6, b_9, r_{10}) & \mapsto (r_1, b_1, r_2, b_2, r_3, b_3, b'_1, r'_2).\\
  \text{Second OR gadget: } (r_4, b_6, r_7, b_7, r_8, b_8, b_5, r_{11}) & \mapsto (r_1, b_1, r_2, b_2, r_3, b_3, b'_1, r'_2),
\end{align*}
with the exception that the segment $\sgt{r_6}{b_5}$ may also play the role of $\sgt{r_1}{b_1}$.

Similarly to an OR gadget, a clause gadget consists of $2^3$ matchings, namely the \emph{$x \vee y \vee z$ matchings} with $x,y,z \in \{0,1\}$. We define the \emph{left}, \emph{middle}, and \emph{right input} of a clause gadget as the left input of the first OR gadget, the right input of the first OR gadget, and the right input of the second OR gadget. We define the \emph{output} of a clause gadget as the output of the second OR gadget.

Note that the middle input segment, i.e., the vertical segment lying in between the two other vertical segments ($\sgt{r_5}{b_5}$ in \Figure~\ref{fig:clauseGadgetbis}), need not be evenly placed between the left input segment and the right input segment. This feature is used to build clause gadgets with non-consecutive variables, such as the topmost clause gadget in \Figure~\ref{fig:MPhi}.

The idea is to have a $0 \vee 0 \vee 0$ clause gadget in $\M_\Phi$ for each clause in $\Phi$.
As we will see next, in the beginning of an untangling sequence starting at $\M_\Phi$, each input may be set to $1$ or may be kept as $0$, changing the $0 \vee 0 \vee 0$ clause gadget into one of the $x \vee y \vee z$ matchings with $x,y,z \in \{0,1\}$.

The following lemma states that the truth table of the logical gate associated with a clause gadget is indeed the expected one.

\begin{figure}[ht]
  \centering
  \includegraphics[scale=\graphicsScale,page=1]{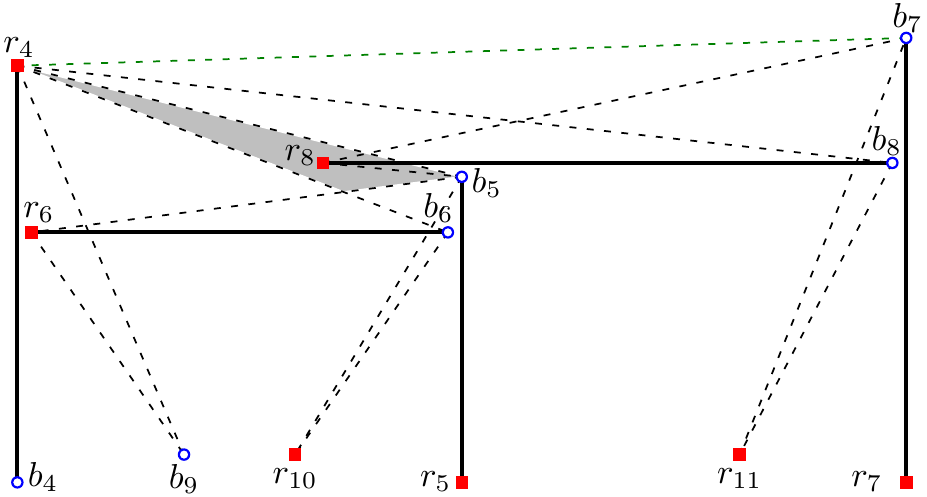}
  \caption{A clause gadget.The $0 \vee 0 \vee 0$ matching is drawn with plain segments.}
  \label{fig:clauseGadgetbis}
\end{figure}

\begin{lemma}
  \label{lem:clausev2}
  The output of a clause gadget is always well defined, and is $0$ if and only if the three inputs of the clause gadget are all $0$.
  More precisely, we have the following.
  \begin{enumerate} 
  \item All the untangle sequences starting at the $0 \vee 0 \vee 0$ matching are of length $4$, and they end at the same matching containing the upper segment $\sgt{r_4}{b_7}$.
  \item All the untangle sequences starting at each of the $x\vee y \vee z$ matchings, where exactly one of $x,y$, or $z$ is $1$, are of length $2$, and they end at matchings excluding the upper segment $\sgt{r_4}{b_7}$.
  \item The unique untangle sequence starting at each of the $x\vee y \vee z$ matchings, where exactly two of $x,y$, and $z$ are $1$, is of length $1$, and it ends at a matching excluding the upper segment $\sgt{r_4}{b_7}$.
  \item The $1 \vee 1 \vee 1$ matching is already crossing free, and it excludes the upper segment $\sgt{r_4}{b_7}$.
  \end{enumerate}
\end{lemma}
\begin{proof}
  It is a consequence of Lemmas~\ref{lem:OR} and of the fact that the OR gadgets are connected so as to not interfere. 
  Indeed, by construction, $r_8$ lies in the inside of the top triangle of the first OR gadget, and is the only overlap between the two OR gadgets. 
  This ensures that all untangle sequences never give rise to an extra crossing that does not already belong to one of the two OR gadgets. 
  In \Figure~\ref{fig:clauseGadgetbis}, we have drawn with dashed line segments all the possible created segments during any possible untangle sequence.
\end{proof}

\paragraph{Padding Gadgets.}
Let $k$ be a non-negative integer. 
A \emph{$k$-padding gadget} triggered by the segment $s$ consists of two matchings built by induction as follows. 

The first matching, denoted $\M_k$, contains $s$ ($s = \sgt{r_4}{b_7}$ in \Figure~\ref{fig:paddingGadget}) and is called the \emph{triggered matching} of the padding gadget ($s$ creates a crossing).
The second matching is called the \emph{non-triggered matching}, and is deduced from the triggered one by removing $s$ (it is crossing free). 

\begin{figure}[ht]
  \centering
  \includegraphics[scale=\graphicsScale,page=2]{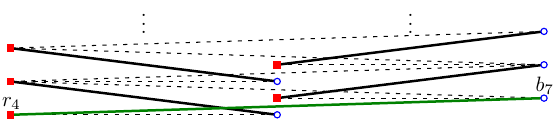}
  \caption{The triggered matching of a padding gadget.}
  \label{fig:paddingGadget}
\end{figure}

If $k=0$, then the triggered matching of a $k$-padding gadget consists of only the segment $s$.
If $k\geq1$, then the triggered matching of a $k$-padding gadget consists of $\M_{k-1}$, the triggered matching of a $(k-1)$-padding gadget, to which we add one new segment crossing only the last created segment of the only untangle sequence starting at $\M_{k-1}$ (\Figure~\ref{fig:paddingGadget}, the dashed segments are all the possible created segments in the unique untangle sequence).

\begin{lemma}
  \label{lem:padding}
  Let $k$ be a non-negative integer.
  There is a unique untangle sequence starting at the triggered matching of a padding gadget, and it is of length $k$.
  The non-triggered matching of a padding gadget is already crossing free.
\end{lemma}
\begin{proof}
  The definition of a $k$-padding gadget yields Lemma~\ref{lem:padding}.
\end{proof}

We complete each clause gadget with a padding gadget in order to penalize a non-satisfied clause by an arbitrary long untangle sequence (\Figure~\ref{fig:clauseGadgetWithPadding}). 
Notice that a padded clause gadget can be arbitrarily scaled and that the position of a clause rectangle is only constrained by the planar embedding of $\Phi$. 

\begin{figure}[ht]
  \centering
  \includegraphics[scale=\graphicsScale,page=3]{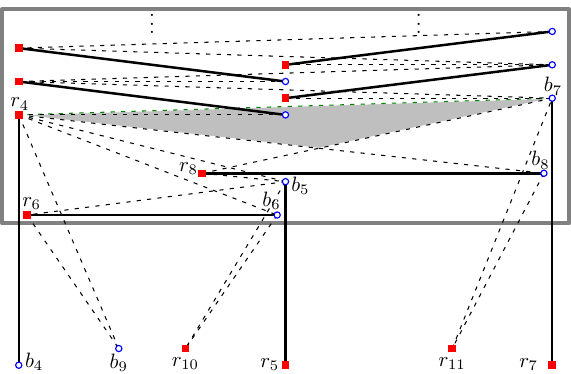}
  \caption{A clause gadget connected to a padding gadget.}
  \label{fig:clauseGadgetWithPadding}
\end{figure}

\paragraph{Matching Computation.}

We now describe, given a planar embedding of $\Phi$, the construction steps of the matching $\M_\Phi$.
Without loss of generality, we only specify the construction of the positive clauses, the construction of the negative clauses being similar.

We need the following definitions for the description.
The \emph{top vertical points} of a positive clause gadget are the topmost endpoints of the vertical segments (e.g. $b_{14}, r_4, r_{12}, b_5, b_7$ in \Figure~\ref{fig:clauseGadget}).
Similarly, the \emph{bottom vertical points} are the bottom endpoints of the same vertical segments (e.g. $r_{14}, b_4, b_{12}, r_5, r_7$ in \Figure~\ref{fig:clauseGadget}).
Let $\overline{p}$ be a top vertical point.
The \emph{horizontal segment of $\overline{p}$} is the horizontal segment lying below $\overline{p}$ which is the closest to $\overline{p}$.
Finally, we define the \emph{substitute point} of $\overline{p}$ as the endpoint of the horizontal segment of $\overline{p}$ which is the closest to $\overline{p}$ (e.g. $r_6$ is the substitute point of $r_4$ in \Figure~\ref{fig:clauseGadget}). 

The construction steps of the matching $\M_\Phi$ are the following.
\begin{enumerate}
  \item Place a clause gadget connected to a $k$-padding gadget in each clause rectangle, and a variable gadget in each variable rectangle, with appropriate scaling.
  \item Connect each clause gadget to its corresponding three variable gadgets with the three vertical segments of the clause gadget aligned with the corresponding vertical edges of the planar embedding of $\Phi$.
  \item Adjust the $x$-coordinates of the vertical segments of each variable gadget to have the top vertical points and the two topmost points of the variable gadget, all in convex position (e.g. in \Figure~\ref{fig:clauseGadget}, $r_{12}$ is on the right of the segment $\sgt{r_4}{b_9}$).
  \item Adjust the $y$-coordinates of the bottom vertical points in the top triangle of each variable gadget so as to place them and the two topmost points of the variable gadget in convex position.
  \item \label{item:adjust} Let $\overline{p}$ be a top vertical point which is not the highest of a variable gadget (e.g. $\overline{p}=r_{12}$ in \Figure~\ref{fig:clauseGadget}).
  Let $\underline{p}$ be the corresponding bottom vertical point (e.g. $\underline{p}=b_{12}$).
  Let $\overline{q}$ be the top vertical point immediately above $\overline{p}$ (e.g. $\overline{q}=r_4$). 
  Let $\underline{q'}$ be the point immediately above $\underline{p}$, taken among the bottom vertical points together with the two topmost points of the variable gadget (e.g. $\underline{q}=b_9$). 
  Adjust the $x$-coordinate of $\tilde{p}$, the substitute point of $\overline{p}$ (e.g. $\tilde{p}=r_{13}$), so that $\tilde{p}$ lies in the triangle $\trgl{\overline{p}}{\underline{q}}{\overline{q}}$ (e.g. a shaded triangle in \Figure~\ref{fig:clauseGadget}; segment $\sgt{r_{13}}{b_{13}}$ must not cross $\sgt{r_4}{b_9}$, but it has to cross $\sgt{r_{12}}{b_9}$).
\end{enumerate}

We have the following lemma.

\begin{lemma}
  \label{lem:matchingComputation}
  Let $\Phi$ be an instance of RPM 3-SAT with $c$ clauses and $v$ variables.
  Let $k$ be a non-negative integer, polynomial in $c$ and $v$.
  The matching  $\M_\Phi$ with $k$-padding gadgets is computed in polynomial time in $c$ and $v$.
\end{lemma}
\begin{proof}
  The number of operations in any execution of these construction steps is linear in $c$ and $v$.
  The coordinates of the points of $\M_\Phi$ are rational numbers with $\OO(\log n)$ bits.  
\end{proof}

\begin{figure}[ht]
  \centering
  \includegraphics[scale=\graphicsScale,page=4]{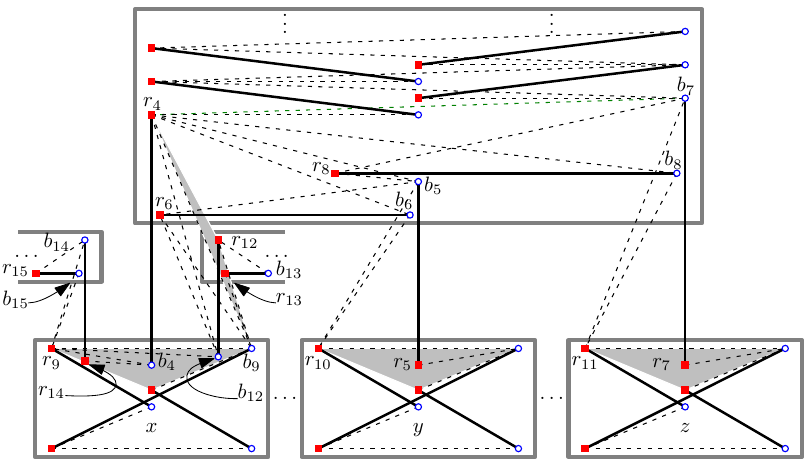}
  \caption{A padded clause gadget connected to $x, y, z$, with branching on $x$.}
  \label{fig:clauseGadget}
\end{figure}

\paragraph{Branching.} 
The following lemma ensures that the connection of multiple vertical segments to a same variable gadget always triggers all the corresponding clause gadgets. 
We start with some definitions.

The set consisting of the top segment of a variable gadget set to true, together with the vertical segments crossing it, and their horizontal segment is called a \emph{branching matching} (such as drawn in \Figure~\ref{fig:sbranching}(b) with plain segments).
The bottom vertical points of a branching matching, listed from left to right, always consist of a certain number, say $a$, of red points followed by a certain number, say $b$, of blue points.
We say that such a branching matching has parameters $a,b$. 
These matchings have the following property.

\begin{lemma}
  \label{lem:branching}
  All the untangle sequences starting at a branching matching with parameters $a,b$ have length $2(a+b)$ and end at the same crossing-free matching (e.g. the segments $\sgt{r_{15}}{b_{14}}, \sgt{r_9}{b_{15}}, \sgt{r_{14}}{b_4}, \sgt{r_4}{b_6}, \sgt{r_6}{b_{12}}, \sgt{r_{12}}{b_{13}}, \sgt{r_{13}}{b_9}$ in \Figure~\ref{fig:sbranching}(b)).
\end{lemma}

\begin{figure}[ht]
  \centering
  (a)\hspace{-1em}\includegraphics[scale=\graphicsScale,page=5]{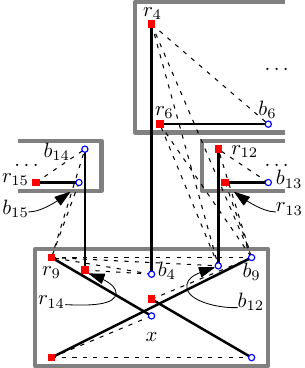}\quad
  (b)\hspace{-1em}\includegraphics[scale=\graphicsScale,page=6]{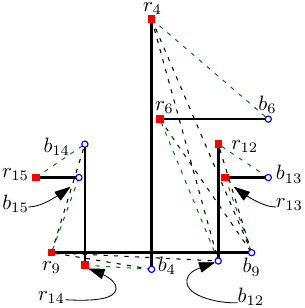}\quad
  (c)\hspace{0.5em}\includegraphics[scale=\graphicsScale,page=7]{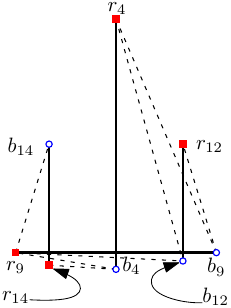}
  \caption{Three views of a branching matching.}
  \label{fig:sbranching}
\end{figure}

\begin{proof}
  First note that a simplified version of this result has been proven in~\cite{BoM16}.
  This simplified version amounts to forget all the horizontal segments, except the top segment of the variable gadget (\Figure~\ref{fig:sbranching}(c)).
  
  It is useful to start by proving this simplified version before Lemma~\ref{lem:branching}.
  We do an induction on $a+b$.
  The base case is trivial, but it provides the possible positions of created segments in the untangle sequences (the dashed segments in \Figure~\ref{fig:sbranching}(c)).
  The inductive case relies on the fact that the points are in convex position.
  Indeed, after any flip, the two created segments play the role of the initial horizontal segment because convex position ensures that any of the non-flipped vertical segments will cross exactly one of the two created segments, and that no extra crossing is created.
  The induction hypothesis then applies on both right and left submatching whose convex hulls are now disjoint.
  
  We now address the issue where each vertical segment is paired with its horizontal segment. 
  Recall that at step~\ref{item:adjust} of the construction of $\M_\Phi$, we have adjusted the $x$-coordinate of $\tilde{p}$, the substitute point of each top vertical point $\overline{p}$, so that $\tilde{p}$ lies in the triangle $\trgl{\overline{p}}{\underline{q}}{\overline{q}}$.
  This ensures that each substitute point can play the role of its corresponding top vertical point from whenever the corresponding horizontal segment has been flipped in an untangle sequence.
\end{proof}

\paragraph{Result.}
The RPM 3-SAT instance being encoded in the matching  $\M_\Phi$, we have the property that the shortest untangling sequence of  $\M_\Phi$ is short if the instance $\Phi$ is satisfiable, and long otherwise. 

\begin{lemma}
  \label{lem:full}
  We have the following case distinction.
  \begin{itemize}
  \item $\Phi$ is satisfiable if and only if there exist untangle sequences starting at $\M_{\Phi}$ which do not trigger any padding gadget, in which case $\dist(\M_{\Phi})$ is at most $v + 5c$.
  \item $\Phi$ is not satisfiable if and only if all untangle sequences starting at $\M_\Phi$ trigger at least one padding gadget, in which case $\dist(\M_{\Phi})$ is at least $v + 7 + k$ where $k$ is arbitrarily large.
  \end{itemize}
\end{lemma}
\begin{proof}
  It is consequence of Lemmas~\ref{lem:var}, \ref{lem:clausev2}, \ref{lem:padding} and \ref{lem:branching}, as we examine the longest possible untangle sequences of $\M_\Phi$ which do not trigger any padding gadget, and the shortest possible untangle sequences of $\M_\Phi$ which trigger at least one padding gadget. 
  In any case, $v$ flips will be performed, one per variable (Lemma~\ref{lem:var}).
  
  In the case where no padding gadget is triggered, the length of the longest possible untangle sequences starting at a clause gadget connected to three variable gadgets is $5$, and is obtained by adding $3$, the length of the untangle sequences of a $0 \vee 0 \vee 1$ matching, and $2$, for the two connections to the negative variables. 
  Counting $5$ flips per clause yields $v + 5c$.
  
  If at least one padding gadget is triggered, this very padding gadget generates $k$ flips. 
  In this case, the length of the shortest possible untangle sequences starting at a clause gadget connected to three variable gadgets and which is known to trigger its padding gadget is $7$, and is obtained by adding $4$, the length of the untangle sequences of a $0 \vee 0 \vee 0$ matching, and $3$, for the three connections to the negative variables.
  All the other cause gadgets may be set to their $1 \vee 1 \vee 1$ matching, adding no flip to the shortest untangle sequence, the length of which is thus $v + 7 + k$.
\end{proof}

We now prove Theorem~\ref{thm:NPapx}, reducing RPM 3-SAT to Problem~\ref{pb:approx}.
Let $\Phi$ be an instance of RPM 3-SAT with $c$ clauses and $v$ variables.
We build the matching $\M_{\Phi}$, which serves as an instance of Problem~\ref{pb:approx}, choosing $k = \alpha (v + 5c) + 1$. 
As $k$ is polynomial in the size of the input ($\alpha$ is a constant), the computation of the matching $\M_{\Phi}$ is polynomial (Lemma~\ref{lem:matchingComputation}).

By hypothesis, we compute an untangle sequence starting at $\M_{\Phi}$ of length $\ell$ at most $\alpha \dist(\M_{\Phi})$.
We decide that $\Phi$ is satisfiable if $\ell \leq \alpha (v+5c)$, and that $\Phi$ is not satisfiable if $\ell > \alpha (v+5c)$.

Indeed, Lemma~\ref{lem:full} ensures the following. 
If $\Phi$ is satisfiable, then the length of the shortest untangle sequence of $\M_{\Phi}$ is at most $v + 5c$. Otherwise the length of the shortest sequence is at least $v+7+k \geq k = \alpha (v+5c) + 1$. This ends the reduction, and proves Theorem~\ref{thm:NPapx}.

\section{Upper Bound on \texorpdfstring{$\dd(n)$}{d(n)}}
\label{sec:algo}

In this section, we prove the following upper bound. 

\begin{theorem}
  \label{thm:algo}
  In the red-on-a-line case, $\dd(n) \leq \binom{n}{2}$.
\end{theorem}

The proof consists of the analysis of the number of flips performed by the recursive algorithm described next.
This analysis is based on a novel approach called \emph{state tracking}.
State tracking is in fact not specific to the red-on-a-line case, which is why Lemma~\ref{lem:tracking} is stated and proven in the non-bipartite setting.
Lemma~\ref{lem:tracking} is then used in the red-on-a-line case to prove Lemma~\ref{lem:Tzone}, which in turn is used to prove Theorme~\ref{thm:algo}.
Lemma~\ref{lem:tracking} also provides an alternative proof of the well-known Theorem~\ref{thm:convexUpperBound}~\cite{BMS19}, which we present at the end of this section.

Throughout, we assume general position (no two blue points with the same $y$-coordinate).
Let the \emph{top} segment of a red-on-a-line matching be the segment with the topmost blue endpoint (\Figure~\ref{fig:topSegments}(a)).

\begin{figure}[!ht]
  \centering
  (a)\includegraphics[page=1,scale=\graphicsScale]{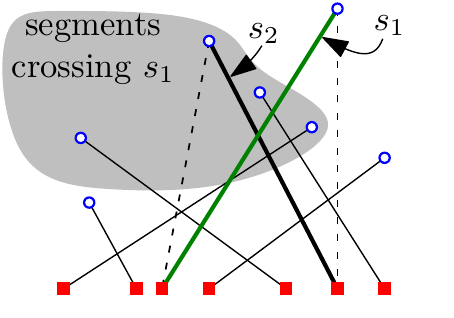}
  \includegraphics[page=16,scale=\graphicsScale]{algo2green}
  (b)\includegraphics[page=15,scale=\graphicsScale]{algo2green}
  \caption{(a) A red-on-a-line matching with $s_1$ as the top segment. 
  (b) The matching just before the first recursive calls of the algorithm, where $s_1$ is free.
  }
  \label{fig:topSegments}
\end{figure}

\paragraph{Algorithm.}

While the top segment $s_1$ of the matching crosses another segment $s_2$, we flip $s_1$ and $s_2$. If multiple segments cross $s_1$, then we choose $s_2$ as the top segment among the segments crossing $s_1$.

The previous loop stops when the top segment $s_1$ has no crossings. At this point, we have that $s_1$ splits the matching into at most two non-empty submatchings, one to each side of $s_1$. We recursively call the algorithm on these submatchings (\Figure~\ref{fig:topSegments}(b)).

\paragraph{Correctness.} 

The next two lemmas prove the correctness of the algorithm. 

\begin{lemma}[\cite{BoM16}]
  \label{lem:convex}
  If a matching admits a partition of submatchings whose convex hulls are all disjoint, then, any sequence of flips in one of the submatchings never affects the other submatchings (\Figure~\ref{fig:hullPartitioning}). 
\end{lemma}
\begin{proof}
  This result can be found in~\cite{BoM16}. Its proof amounts to the observation that the flip operation leaves the convex hull unchanged (in \Figure~\ref{fig:hullPartitioning}, the dashed segments are the results of possible flip sequences).
\end{proof}

\begin{figure}[!ht]
  \centering
  \includegraphics[scale=\graphicsScale]{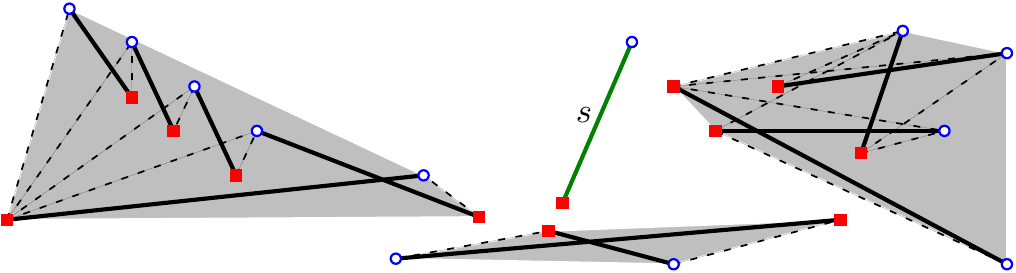}
  \caption{A partition of $4$ submatchings whose convex hulls are all disjoint. The segment $s$ is the only free segment.}
  \label{fig:hullPartitioning}
\end{figure}

We say that a segment $s$ is \emph{free} if the matching admits a partition of submatchings whose convex hulls are all disjoint, and one of the submatchings consists of the segment $s$ alone. 
In \Figure~\ref{fig:hullPartitioning}, the segment $s$ is the only free segment.

\begin{lemma}
  \label{lem:free}
  The algorithm always makes the top segment free before recursive calls.
\end{lemma}
\begin{proof}
  The algorithm repeats the flip step until the top segment is free. 
  As any sequence of flips is finite, this eventually happens.
  The recursive calls of the algorithm happen when and only when the top segment is free. 
\end{proof}

The correctness of the algorithm follows from Lemma~\ref{lem:convex} and~\ref{lem:free}.

\paragraph{Flip Complexity.}

The analysis of the number of flips performed by the algorithm stems from the following observations.
We define three possible \emph{states} for a pair of segments (\Figure~\ref{fig:states}).
\begin{itemize}
    \item State $\X$: the segments are crossing.
    \item State $\HH$: the segments are not crossing and their endpoints are in convex position.
    \item State $\T$: the endpoints are not in convex position. 
\end{itemize}
\begin{figure}[!ht]
  \centering
  \includegraphics[scale=\graphicsScale]{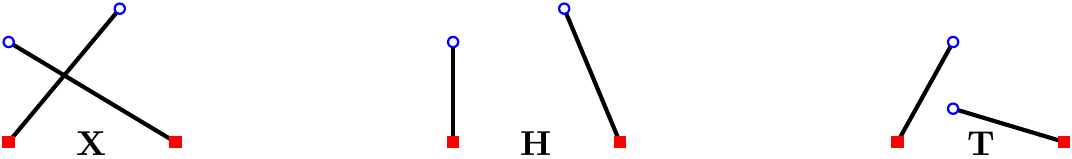}
  \caption{The three different states of pairs of segments.}
  \label{fig:states}
\end{figure}

In the convex case, there are no $\T$-states and a flip increases the number of $\HH$-pairs by at least $1$ unit, and decreases the number of $\X$-pairs as well. Hence, counting either $\X$ or $\HH$-pairs yields the $\binom{n}{2}$ upper bound on $\D(n)$ (this upper bound is in~\cite{BMS19} and the alternative proof we mentioned is made precise in Theorem~\ref{thm:convexUpperBound}).
However, when the points are not in convex position, counting 
$\HH$ and $\X$-pairs is fundamentally different. We will see that counting $\HH$-pairs is more useful to prove the desired bounds.

When the points are not in convex position, a flip may decrease the number of $\HH$-pairs.
\Figure~\ref{fig:HT} shows two such situations where flipping $\pair{s_1}{s_2}$ does not increase the number of $\HH$-pairs. There is one $\HH$-pair involving segment $s$ before the flip, and none after the flip. Notice that, if we added multiple segments close to $s$, the number of $\HH$-pairs would actually decrease.
However, the algorithm avoids these situations by choosing to flip top segments. 
The full proof involves state tracking, a novel approach to analyze flip sequences, which is described next.

\begin{figure}[!ht]
  \centering
  \includegraphics[scale=\graphicsScale,page=1]{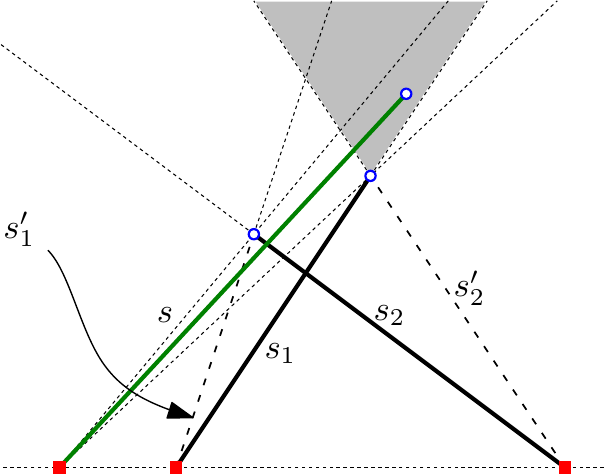}\qquad\qquad
  \includegraphics[scale=\graphicsScale,page=2]{THStateTracking}\\ \vspace{\baselineskip}
    (a) \quad 
     \begin{tabular}{r@{: }c}
        $\pair{\green{s}}{s_1}$ & $\HH$ \\
        $\pair{\green{s}}{s_2}$ & $\X$ \\
        \grey{$\pair{s_1}{s_2}$} & \grey{$\X$}
      \end{tabular}
      \quad
      \begin{tabular}{r@{: }c}
        $\pair{\green{s}}{s_2'}$ & $\T$ \\
        $\pair{\green{s}}{s_1'}$ & $\X$ \\
        \grey{$\pair{s_1'}{s_2'}$} & \grey{$\HH$}
      \end{tabular}
        \hspace{0.14\linewidth}
    (b) \quad 
      \begin{tabular}{r@{: }c}
        $\pair{\green{s}}{s_1}$ & $\HH$ \\
        $\pair{\green{s}}{s_2}$ & $\T$ \\
        \grey{$\pair{s_1}{s_2}$} & \grey{$\X$}
      \end{tabular}
      \quad 
      \begin{tabular}{r@{: }c}
        $\pair{\green{s}}{s_1'}$ & $\T$ \\
        $\pair{\green{s}}{s_2'}$ & $\T$ \\
        \grey{$\pair{s_1'}{s_2'}$} & \grey{$\HH$}
      \end{tabular}
  \caption{Two cases where flipping $\pair{s_1}{s_2}$ does not increase the number of $\HH$-pairs. The upper cone of $s_1,s_2'$ is shaded.}
  \label{fig:HT}
\end{figure}

\paragraph{State Tracking.}
We have $\binom{n}{2}$ pairs of segments before and after a flip. Each pair has an associated state. However, since two segments change in the matchings, there is no clear correspondence between the state of each pair before and after the flip. State tracking establishes this correspondence by making choices of which pair of segments in the initial matching corresponds to which pair of segments in the resulting matching. These choices are performed deliberately to obtain certain state transitions instead of others and prove the desired bounds.

The following notations will be used throughout the rest of this section and are summarized in \Figure~\ref{fig:stateTrackingFlip}. 
Let $r_1, r_2$ be two red points and $b_1, b_2$ be two blue points. 
Let $s_1, s_2, s_1', s_2'$ be the following four segments respectively: $\sgt{r_1}{b_1}$, $\sgt{r_2}{b_2}$, $\sgt{r_1}{b_2}$, $\sgt{r_2}{b_1}$.
We consider a flip that replaces the pair of segments $\pair{s_1}{s_2}$ by $\pair{s_1'}{s_2'}$. 
Let $\M$ denote the matching before the flip and $\M'$ denote the resulting matching after the flip.

\begin{figure}[!ht]
  \centering
  \includegraphics[scale=\graphicsScale,page=1]{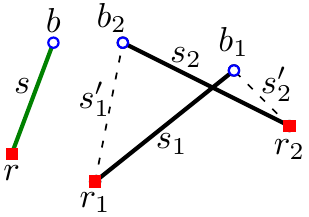}
  \caption{Notations for a generic flip and for a variable segment $s$.}
  \label{fig:stateTrackingFlip}
\end{figure}

We order the $\binom{n}{2}$ pairs of segments of $\M$ in a column vector. 
There are three types of pairs of segments in $\M$ with respect to the flip: the \emph{unaffected} pairs (involving neither $s_1$ nor $s_2$), the \emph{flipping} pair $\pair{s_1}{s_2}$, and the \emph{affected} pairs (involving exactly one of $s_1$ or $s_2$). 
We choose the new order of the $\binom{n}{2}$ pairs of segments of $\M'$ in a way that satisfies the following properties with respect to the previous vector. 
The unaffected pairs keep the same indices.
The pair $\pair{s_1'}{s_2'}$ gets the index of $\pair{s_1}{s_2}$. Next, we describe the remaining indices.

Let $s$ be a segment of $\M$ distinct from $s_1$ and $s_2$. 
Let $r$ and $b$ be the red and blue endpoints of $s$.
Let $i_1$ and $i_2$ be the indices of $\pair{s}{s_1}$ and $\pair{s}{s_2}$, and let $\gS_1$ and $\gS_2$ be their respective states.
Let $\gS_1'$ and $\gS_2'$ be the respective states of $\pair{s}{s_1'}$ and $\pair{s}{s_2'}$.
We restrict our choice to the following two options: 
\begin{itemize}
 \item index $\pair{s}{s_1'}$ with $i_1$, and $\pair{s}{s_2'}$ with $i_2$, or
 \item index $\pair{s}{s_1'}$ with $i_2$, and $\pair{s}{s_2'}$ with $i_1$. 
\end{itemize}
We call such a choice a \emph{tracking choice}. 
We say that a pair of segments in $\M$ \emph{turns into} a pair in $\M'$ when they have the same index. 
We denote $\gS \to \gS'$ to specify that the pairs of segments with a given index go from the state $\gS$ to the state $\gS'$. 
In the following, we use $\gS_1 \gS_2 \to \gS_1' \gS_2'$ as a shorthand notation to say that we have the two following tracking choices: either $\gS_1\to \gS_1'$ and $\gS_2 \to \gS_2'$ or $\gS_1\to \gS_2'$ and $\gS_2 \to \gS_1'$.

There are $3^2$ possible such \emph{transitions} $\gS \to \gS'$.
Yet, the next two lemmas ensure that some transitions can be ruled out by tracking choices.
Lemma~\ref{lem:tracking} actually holds for any (possibly non-bipartite) matching, while Lemma~\ref{lem:Tzone} is specific to the red-on-a-line case.
Both lemmas are proved analyzing the tracking choices of each possible position of a segment $s$ relatively to the flipping pair.

\begin{lemma}
  \label{lem:tracking}
  There always exists a tracking choice avoiding the $\HH \to \X$ transition.
\end{lemma}
\begin{proof}
There clearly exists a tracking choice avoiding the $\HH \to \X$ transition unless we have either a transition (i) $\HH  \HH \to \X  \gS$ or (ii) $\HH  \gS \to \X  \X$, where $\gS \in \{ \X, \HH, \T \}$. We show that these two cases are not possible.

(i) $\HH  \HH \to \X  \gS$: If both the pairs $s,s_1$ and $s,s_2$ are $\HH$ while at least one of the two pairs $s,s'_1$ and $s,s'_2$ is $\X$, then the final $\X$ state implies that $s$ crosses $s_1$ or $s_2$, which contradicts the two initial $\HH$ states.    

(ii) $\HH  \gS \to \X  \X$: If one of the two pairs $s,s_1$ and $s,s_2$ is $\HH$ while both pairs $s,s'_1$ and $s,s'_2$ are $\X$, then the two final $\X$ states imply that $s$ crosses $s'_1$ and $s'_2$. It follows that $s$ also crosses $s_1$ and $s_2$, which is again a contradiction.
\end{proof}

\paragraph{State Tracking in the Red-on-a-Line Case.}

\Figure~\ref{fig:stateTrackingMapNotations} summarizes the notations for a generic red-on-a-line flip and an variable segment $s$.
\Figures~\ref{fig:stateTrackingMapCases}, \ref{fig:stateTrackingMapCase1}, \ref{fig:stateTrackingMapCase2}, and \ref{fig:stateTrackingMapCase3.1} then provide ``maps'' of essentially all the possible situations of tracking choices in the red-on-a-line case. 
These figures are used to prove the next lemma.

  \begin{figure}[!ht]
  \centering
  \includegraphics[scale=\graphicsScale,page=1]{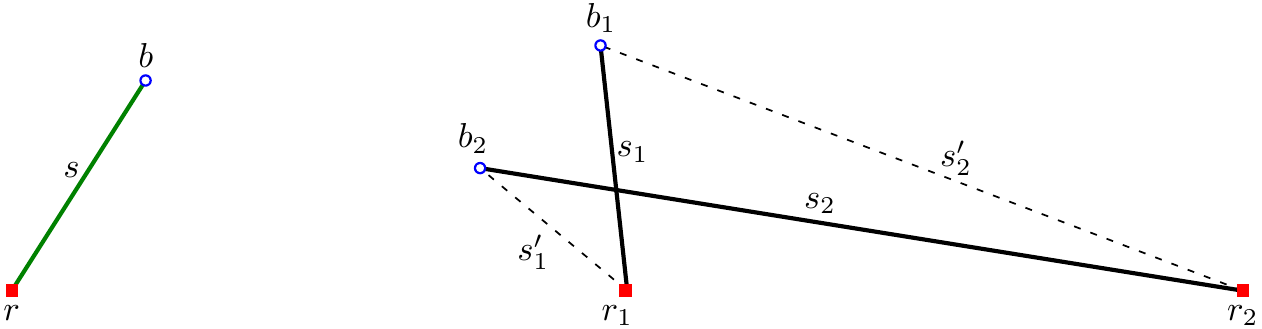}
  \caption{Notations used in \Figures~\ref{fig:stateTrackingMapCases}, \ref{fig:stateTrackingMapCase1}, \ref{fig:stateTrackingMapCase2}, and \ref{fig:stateTrackingMapCase3.1} for a generic red-on-a-line flip and an variable segment $s$.}
  \label{fig:stateTrackingMapNotations}
  \end{figure}

  \begin{figure}[!ht]
  \centering
  \includegraphics[scale=\graphicsScale,page=2]{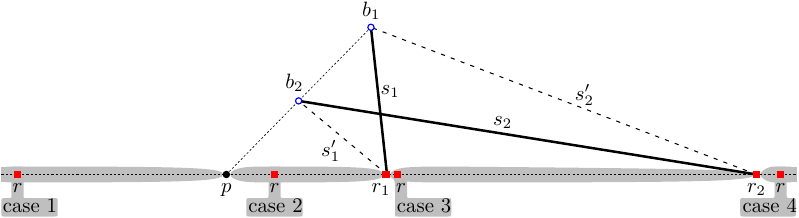}
  \caption{The four possible cases for the position of $r$.}
  \label{fig:stateTrackingMapCases}
  \end{figure}

  \begin{figure}[p]
  \centering
  \includegraphics[scale=\graphicsScale,page=3]{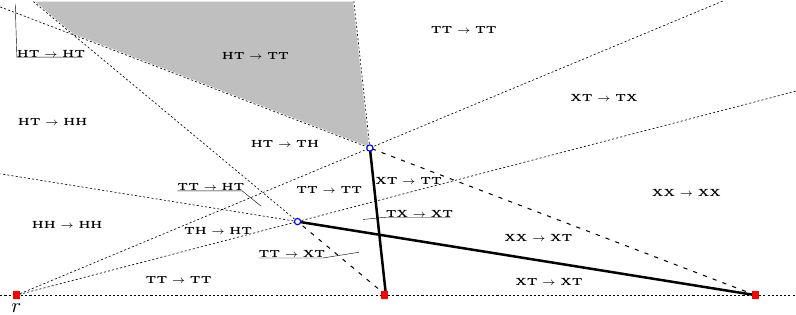}
  \caption{The case 1 ``map'' of all the possible red-on-a-line tracking choices. Tracking choices cannot avoid the transition $\HH \to \T$ in the shaded region.}
  \label{fig:stateTrackingMapCase1}
  \end{figure}
  
  \begin{figure}[p]
  \centering
  \includegraphics[scale=\graphicsScale,page=4]{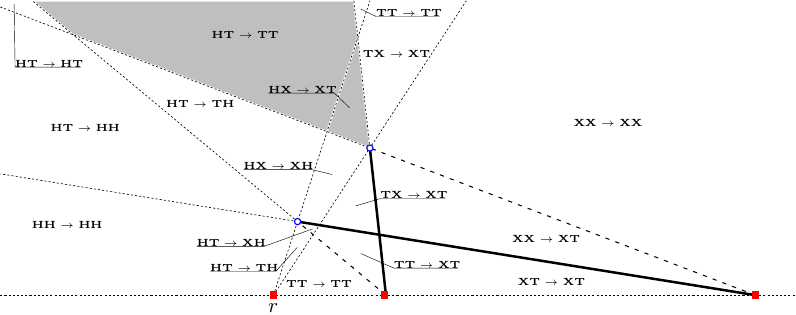}
  \caption{The case 2 ``map'' of all the possible red-on-a-line tracking choices. Tracking choices cannot avoid the transition $\HH \to \T$ in the two shaded regions.}
  \label{fig:stateTrackingMapCase2}
  \end{figure}
  
  \begin{figure}[p]
  \centering
  \includegraphics[scale=\graphicsScale,page=5]{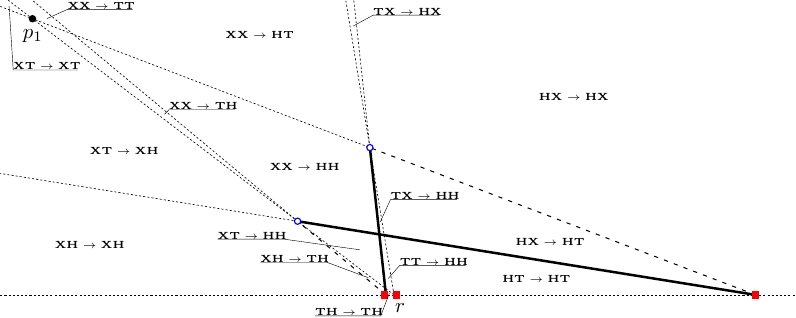}
  \caption{The case 3.1 ``map'' of all the possible red-on-a-line tracking choices.}
  \label{fig:stateTrackingMapCase3.1}
  \end{figure}

  \Figures~\ref{fig:stateTrackingMapCase1}, \ref{fig:stateTrackingMapCase2}, and \ref{fig:stateTrackingMapCase3.1} are generated by a brute force computation of the states $\gS_1, \gS_2, \gS_1', \gS_2'$ of the four pairs $\pair{s}{s_1}$, $\pair{s}{s_2}$, $\pair{s}{s_1'}$, $\pair{s}{s_2'}$ for each position case for $r$ (\Figure~\ref{fig:stateTrackingMapCases}) and for each position case for $b$ (in \Figures~\ref{fig:stateTrackingMapCase1}, \ref{fig:stateTrackingMapCase2}, and \ref{fig:stateTrackingMapCase3.1}, each cell of the arrangement of lines corresponds to a position case for $b$). 
  In the following, we make sure that no case is forgotten.
  
  We assume, without loss of generality, that $r_1$ is on the left of $r_2$, and that $b_1$ is higher than $b_2$. 
  Let $p$ be the intersection between the line $\lineT{b_1}{b_2}$ and the red-point line.
  There are, indeed, four possible open intervals for the position of $r$ on the red-point line: $\opnint{-\infty}{p}$, $\opnint{p}{r_1}$, $\opnint{r_1}{r_2}$, and $\opnint{r_2}{\infty}$ (\Figure~\ref{fig:stateTrackingMapCases}). 
  This yields four cases, respectively. 
  We do not explicitly describe case~4 as it is similar to case~2. Indeed, case~2 and case~4 map to each other by exchanging the labels of $r_1$ and $r_2$, as well as $b_1$ and $b_2$. The fact that the point $p$ is still on the left of $r_1$ and $r_2$ is not a problem since we are studying incidence proprieties. Another way to see it, is to consider the projective plane.
    
  As we have assumed the blue points to lie in the upper half-plane, these four cases split further into sub-cases. 
  However, no-loss-of-generality assumptions and symmetries simplify the analysis.
  Without loss of generality, we first assume that the lines $\lineT{r_1}{b_2}$ and $\lineT{r_2}{b_1}$ intersect in the upper half-plane, as it will only generate more cells to the upper part of the arrangement of lines. 
  
  Second, we examine case~3. 
  Let $p_1$ be the intersection of the lines $\lineT{r}{b_2}$ (see \Figure~\ref{fig:stateTrackingMapCase3.1}) and $\lineT{r_2}{b_1}$, and $p_2$ be the intersection of the lines $\lineT{r}{b_1}$ and $\lineT{r_1}{b_2}$. 
  Case~3 decomposes into:
  \begin{itemize}
      \item case~3.1 where $p_1$ lies in the upper half-plane and $p_2$ in the lower, 
      \item case~3.2 where both $p_1$ and $p_2$ lie in the upper half-plane, 
      \item case~3.3 where $p_1$ lies in the lower half-plane and $p_2$ in the upper, and 
      \item case~3.4 where both $p_1$ and $p_2$ lie in the lower half-plane.
  \end{itemize}
  Cases~3.1 and~3.3 are similar, while case~3.2 is just a superposition of both of them. 
  More precisely, when compared to case~3.4, the extra cell of the arrangement generated by case~3.1 (the cell in the top left corner of \Figure~\ref{fig:stateTrackingMapCase3.1}) corresponds to the possible tracking choices summarized by the notation $\X \T \to \X \T$. 
  Similarly, the extra cell generated by case~3.3 corresponds to $\T \X \to \T \X$. 
  The two extra cells generated by case~3.2 are the same as the two previous ones.
  We thus assume case~3.1 (as it is easier to draw in our setting) without loss of generality.
  All these assumptions made, the remaining cases now corresponds to \Figures~\ref{fig:stateTrackingMapCase1}, \ref{fig:stateTrackingMapCase2}, and \ref{fig:stateTrackingMapCase3.1}.

The next lemma is similar to Lemma~\ref{lem:tracking}, but specific to red-on-a-line matchings.  
We will use it to additionally avoid the $\HH \to \T$ transition.
To state Lemma~\ref{lem:Tzone}, we define the \emph{upper cone} of two segments $\sgt{r_3}{b_3}, \sgt{r_4}{b_3}$ as the locus of the points that are separated from the horizontal line $\lineT{r_3}{r_4}$ by the two lines $\lineT{r_3}{b_3}$ and $\lineT{r_4}{b_3}$ (\Figure~\ref{fig:upperConeAndRay}(a)).
We also define the \emph{upper ray} of a segment as the open ray with the blue point as its origin, the segment as its direction, and going upwards (\Figure~\ref{fig:upperConeAndRay}(b)).

\begin{figure}[!ht]
  \centering
  (a) \quad \includegraphics[scale=\graphicsScale]{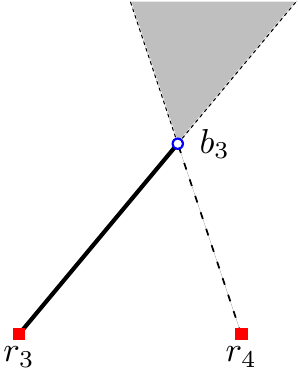}
  \hspace{10em}
  (b) \quad \includegraphics[scale=\graphicsScale]{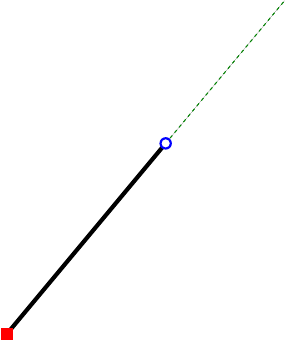}
  \caption{\label{fig:upperConeAndRay}(a) The upper cone of $\sgt{r_3}{b_3}$ and $\sgt{r_4}{b_3}$ is shaded.  (b) The upper ray of the segment is dotted. 
  }
\end{figure}

\begin{lemma}
  \label{lem:Tzone}
  In the red-on-a-line case, if the blue point $b$ of $s$ is not in any of the two upper cones of $s_1, s_2'$ and $s_2, s_1'$, then
   there always exists a tracking choice that avoids $\HH \to \T$ for the pairs $\pair{s}{s_1}$ and $\pair{s}{s_2}$ while still avoiding $\HH \to \X$.
\end{lemma}
\begin{proof}
  First, we check that there are only two possible upper cones defined by two segments of $s_1, s_2, s_1', s_2'$. 
  Indeed, only two pairs among them have a common blue point.

  Then, we note that, for $\pair{s}{s_1}$ or $\pair{s}{s_2}$ to be in state $\HH$, the red point $r$ of $s$ cannot be between $r_1$ and $r_2$, the red points of $s_1$ and $s_2$.
  Without loss of generality, we assume $r$ to lie on the left side of $r_1$ and $r_2$.
  
  For $\pair{s}{s_1'}$ or $\pair{s}{s_2'}$ to be in state $\T$, $s$ has to cross at least one of the upper rays of $s_1'$ or $s_2'$.
  
  The only two combinations of states for $\{\pair{s}{s_1}, \pair{s}{s_2}\}$ and $\{\pair{s}{s_1'}, \pair{s}{s_2'}\}$ which do not leave us the choice to avoid the $\HH \to \T$ transition are $\{\HH, \T\}$ and $\{\T, \T\}$, and $\{\HH, \X\}$ and $\{\T, \X\}$. 
  In any case, $b$ must be in the right most of the two upper cones of segments $s_1, s_2, s_1', s_2'$. 
  More precisely, $b$ lies in one of the three shaded regions of \Figures~\ref{fig:stateTrackingMapCase1} and \ref{fig:stateTrackingMapCase2}. 
  These three shaded regions also correspond to \Figure~\ref{fig:HT} where case~1 is omitted but similar.
  The other cases are either not feasible geometrically, or with a possibility to make tracking choices so as to avoid transition $\HH \to \T$. 
\end{proof}

\paragraph{Proof of Theorem~\ref{thm:algo}.}

We are now ready to prove Theorem~\ref{thm:algo}.

\begin{proof} 
   Let $\nbf(\M)$ be the total number of flips performed by the algorithm on an $n$-segment input matching $\M$ and let $\nbg(\M)$ be the number of flips performed by the algorithm before the recursive calls. Let $\M_r$ denote the matching before the recursive calls.
   The recursive calls take two submatchings of $\M_r$ that we call $\M_1$ and $\M_2$, yielding the following recurrence relation.
   \[
   \nbf(\M) = \nbf(\M_1) + \nbf(\M_2) + \nbg(\M)
   \]
   
   Let $\nbh(M)$ be the number of $\X$-pairs plus the number of $\T$-pairs in a matching $M$, that is, the number of pairs that are not $\HH$-pairs. 
   Lemma~\ref{lem:Tzone} ensures that 
   \[\nbg(\M) \leq \nbh(\M) - \nbh(\M_r) \leq \nbh(\M) - \nbh(\M_1) - \nbh(\M_2).\]
       
   Clearly, $\nbf(\emptyset) = 0$.
   We suppose that, for all $\M'$ with less than $n$ segments, we have
   $\nbf(\M') \leq \nbh(\M')$.
   Then by induction we get
   \[
   \nbf(\M) \leq \nbh(\M_1) + \nbh(\M_2) + \nbh(\M) - \nbh(\M_1) - \nbh(\M_2) =  \nbh(\M).
   \] 
   
   Theorem~\ref{thm:algo} follows since $\nbh(\M) \leq \binom{n}{2}$. 
\end{proof}  

\paragraph{State Tracking in the Convex Case.}

State tracking also applies to the widely studied convex case, providing a more conceptual proof of the following theorem from~\cite{BMS19}. 
Even though we will not use this well-known result, we may as well state it.
This theorem actually holds for any straight-line non-bipartite perfect matching.

\begin{theorem}[Theorem 5 of \cite{BMS19}]
  \label{thm:convexUpperBound}
  In the convex case, any untangle sequence is of length at most $\binom{n}{2}$.
\end{theorem}
\begin{proof}
  In the convex case, the $\T$-state does not exist. Lemma~\ref{lem:tracking} thus ensures that the number of $\HH$-pairs increases of at least $1$ unit at each flip.
\end{proof}

\section{Upper Bound on \texorpdfstring{$\D(n)$}{D(n)}}
\label{sec:upperB}

In this section we prove the following theorem.

\begin{theorem}
  \label{thm:upperB}
  In the red-on-a-line case, $\D(n) \leq \binom{n}{2}\frac{n+4}{6}$.
\end{theorem}

To prove Theorem~\ref{thm:upperB}, we define a potential function $\Phi$ that maps a red-on-a-line matching to an integer from $0$ to $\binom{n}{2} \frac{n+4}{3}$. 
Since $\Phi$ decreases by at least $2$ units at each flip, the theorem follows. 
We first give the definitions needed to present $\Phi$. 
Then, we prove four lemmas yielding Theorem~\ref{thm:upperB}. 

Let $\M$ be a red-on-a-line matching.
Let $r_1, \dots, r_n$ be the red points, from left to right.
Let $\lt$ be a line, parallel to the line of the red points and above all the points. 
For each $k$ in $\{1, \dots, n\}$, we project the blue points onto $\lt$, using $r_k$ as a focal point. 
More precisely, each blue point $b$ maps to a point $t_k(b)$, the intersection between the ray $\ray{r_k}{b}$ and the line $\lt$ (\Figure~\ref{fig:projectionk}(a)). 
We also define the function $t_k$ of a red-blue segment $\sgt{r}{b}$ as the segment $t_k(\sgt{r}{b}) = \sgt{r}{t_k(b)}$ (\Figure~\ref{fig:projectionk}(b)). 

\begin{figure}[!ht]
  \centering
  (a)\hspace{-1em}\includegraphics[scale=\graphicsScale,page=5]{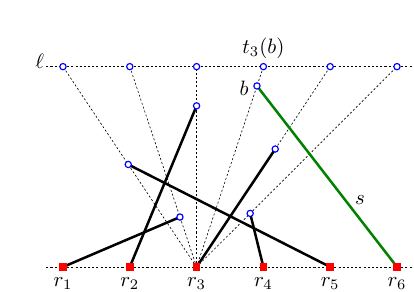}\qquad
  (b)\hspace{-1em}\includegraphics[scale=\graphicsScale,page=6]{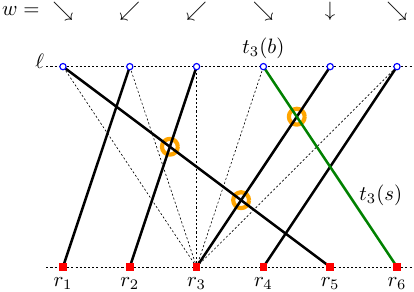}\\
  \caption{(a) The projection $t_k$ for $k=3$. (b) The  segments $t_3(\cdot)$. The three $3$-observed crossing $3$-pairs are circled.}
  \label{fig:projectionk}
  \label{fig:kcrossingkpairs}
\end{figure}

We may abbreviate a pair of segments $\pair{\sgt{r_i}{b}}{ \sgt{r_j}{b'}}$ as $\ipair{i}{j}$ when the points $b$ and $b'$ can be deduced from the underlying matching.
Let $k$ be an integer in $\{1, \dots,n\}$. 
We say that two segments are \emph{$k$-observed crossing} if the extended projection $t_k(\cdot)$ maps them to crossing segments (\Figure~\ref{fig:projectionk}(b)).
A pair of segments $\ipair{i}{j}$ 
is a \emph{$k$-pair} if $i \leq k \leq j$. 
A \emph{$k$-flip} is then a flip of a $k$-pair.
We have the following lemma.

\begin{lemma}
  \label{lem:kcrossing}
  A crossing $k$-pair is necessarily $k$-observed crossing. 
\end{lemma}
\begin{proof}
  Let $\pair{\sgt{r_i}{b}}{ \sgt{r_j}{b'}}$ be a crossing $k$-pair.
  We suppose, without loss of generality, that $i<j$ (e.g. $i=2$, $k=3$, and $j=5$ in \Figure~\ref{fig:projectionk}). 
  
  The fact that the $k$-pair $\pair{\sgt{r_i}{b}}{ \sgt{r_j}{b'}}$ is crossing means that the four points are in convex position, and that they appear as $r_i, r_j, b, b'$ on their convex hull in counter-clockwise order.
  Since $i \leq k \leq j$, the point $r_k$ is also on the boundary of the convex hull of the four points.
  Therefore, the projection $t_k(\cdot)$ will not change the convex-hull order and the segments $\sgt{r_i}{t_k(b)}$ and $\sgt{r_j}{t_k(b')}$ will cross. 
\end{proof}

We define $\Phi_k(\M)$, the \emph{$k$-th potential} of $\M$, as the number of $k$-observed crossing $k$-pairs (\Figure~\ref{fig:kcrossingkpairs}(b)).
Lemma~\ref{lem:upperB1} shows that the $k$-th potential $\Phi_k$ is at most $(k-1)(n-k) + n-1$. 
Lemma~\ref{lem:upperB2} shows that $\Phi_k$ never increases, and decreases by at least $1$ unit at each $k$-flip. 

\begin{lemma}
  \label{lem:upperB1}
  The $k$-th potential $\Phi_k$ takes integer values from $0$ to $k(n+1) -k^2 - 1$.
\end{lemma}
\begin{proof}
  The $k$-th potential $\Phi_k(\M)$ is at most the number of $k$-pairs in $\M$, crossing or not.
      There are exactly $(k-1)(n-k)$ $k$-pairs of the form $\ipair{i}{j}$ with $i < k < j$.
      There are exactly $k-1$ $k$-pairs of the form $\ipair{i}{k}$ with $i < k$. 
      There are exactly $n-k$ $k$-pairs of the form $\ipair{k}{j}$ with $k < j$.
  In total, there are $k(n+1) -k^2 - 1$ $k$-pairs in $\M$. 
\end{proof}

\begin{lemma}
  \label{lem:upperB2}
  The $k$-th potential $\Phi_k$ never increases, and decreases by at least $1$ unit at each $k$-flip.
\end{lemma}

\begin{proof}
  We order the projected blue points on $\ell$ from left to right.
  We then map each projected blue point $t_k(b)$ to an element in $\{\lft,\ctr,\rgt\}$: 
  \begin{itemize}
      \item $t_k(b)$ is mapped to $\lft$ if $b$ is matched to a red point on the left of $r_k$, 
      \item $t_k(b)$ is mapped to $\ctr$ if $b$ is matched to $r_k$, 
      \item $t_k(b)$ is mapped to $\rgt$ if $b$ is matched to a red point on the right of $r_k$. 
  \end{itemize}
  Let $w=w_1 \dots w_n$ be the word on the alphabet
  $\{\lft,\ctr,\rgt\}$
  induced by the order of the projected blue points and the map. 
  For instance, in Figure~\ref{fig:kcrossingkpairs} with $k=3$, $w = \rgt\lft\lft\rgt\ctr\rgt$.
  
  Let the total order of the symbols be $\lft \; \prec \; \ctr \; \prec \; \rgt$.
  An \emph{inversion} in $w$ is a pair $\pair{w_i}{w_j}$ with $i<j$ and $w_j \prec w_i$. 
  The inversions in $w$ are in bijection with the $k$-observed crossing $k$-pairs in $\M$.
  Thus, by definition, $\Phi_k(\M)$ is the number of inversions in $w$. 
  Lemma~\ref{lem:upperB2} follows from the following two observations.
  
  (i) Any flip which is not a $k$-flip swaps two $\lft$ or two $\rgt$ in $w$, resulting in word $w'$ identical to $w$.
  
  (ii) Lemma~\ref{lem:kcrossing} ensures that a crossing $k$-pair corresponds to an inversion in $w$. 
  Thus, a $k$-flip exchanges the two symbols of an inversion in $w$, resulting in word $w'$ with at least one inversion less than in $w$.
\end{proof}

We now define $\Phi(\M)$, the \emph{potential} of $\M$, as the sum of $\Phi_k(\M)$, for $k$ in $\{1, \dots, n\}$. The following lemma presents the key properties of $\Phi$.

\begin{lemma}
  \label{lem:upperB3}
  The potential $\Phi$ takes integer values from $0$ to $
  \binom{n}{2} \frac{n+4}{3}$, and decreases by at least $2$ units at each flip.
\end{lemma}
\begin{proof}
  We know that $\Phi$ takes non-negative integer values by definition and, by Lemma~\ref{lem:upperB1}, an upper bound on $\Phi$ is
  \begin{align*}
    \sum_{k=1}^{n} \left( k(n+1) -k^2 - 1 \right) 
    & = (n+1)\sum_{k=1}^{n} k - \sum_{k=1}^{n} k^2 - n \\
    & = (n+1)\frac{n(n+1)}{2} - \frac{n(n+1)(2n+1)}{6} - n \\
    & = \frac{n}{6} ( n^2 + 3n - 4 ) \\
    & = \binom{n}{2} \frac{n+4}{3}.
  \end{align*}
  
  Finally, Lemma~\ref{lem:upperB2} ensures that $\Phi$ decreases by at least $2$ units at each flip. 
  Indeed, a flip of a pair $\ipair{i}{j}$ is counted at least twice: once in $\Phi_i$ as an $i$-flip, and once in $\Phi_j$ as a $j$-flip.
\end{proof}

Theorem~\ref{thm:upperB} follows from Lemma~\ref{lem:upperB3}.

\section{Lower Bounds}
\label{sec:lowerB}

In this section, we prove the following two lower bounds. 

\begin{theorem}
  \label{thm:lowerBD}
  In the red-on-a-line case, for even $n$, $\D(n) \geq \frac{3}{2}\binom{n}{2} - \frac{n}{4}$.
\end{theorem}

\begin{theorem}
  \label{thm:lowerBd}
  In the convex case, for even $n$, $\dd(n) \geq \frac{3n}{2} - 2$. 
\end{theorem}

To prove Theorem~\ref{thm:lowerBD}, it suffices to present a long untangle sequence. The initial matching of the sequence is represented in \Figure~\ref{fig:butterfly}(a).
To prove Theorem~\ref{thm:lowerBd}, we need to show that every untangle sequence starting at a given configuration (represented in \Figure~\ref{fig:butterfly}(b)) is long enough. We do so by showing that every flip reduces the number of crossings by exactly one unit.

\begin{figure}[!htb]
  \centering
  \quad(a) \includegraphics[scale=\graphicsScale,page=1]{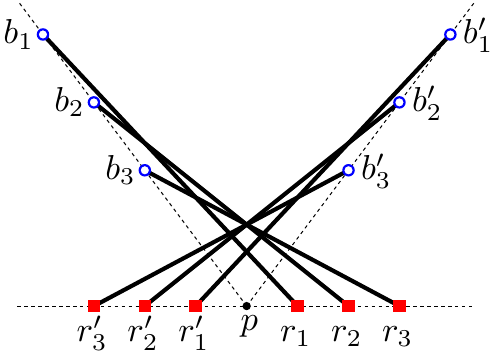}\qquad\qquad 
  (b) \includegraphics[page=4,scale=\graphicsScale]{fence}\\
  \caption{(a) A $3$-butterfly to lower bound $\D(6)$. (b) A $5$-fence to lower bound $\dd(10)$.}
  \label{fig:butterfly}
  \label{fig:fence}
\end{figure}

\subsection{Lower Bound on \texorpdfstring{$D(n)$}{D(n)}}
\label{sec:butterflyLowerB}

We provide a $2m$-segment red-on-a-line matching which we call an \emph{$m$-butterfly}. 
There exists an untangle sequence starting at an $m$-butterfly of length $\frac{3}{2}\binom{2m}{2} - \frac{m}{2}$. 
Next, we give the precise definition of an $m$-butterfly  and some of its important properties.
Then, we give some intuition of how to come up with an untangle sequence longer than the number of pairs of segments. 
Finally, we prove that there exists an untangle sequence starting at an $m$-butterfly of length $\frac{3}{2}\binom{2m}{2} - \frac{m}{2}$ with two lemmas. 

\paragraph{Butterfly.}

For an integer $m$, we define an \emph{$m$-butterfly} as the following matching with $n=2m$ segments.
For $i$ from $1$ to $m$ we have red points $r_i = (i/(m+1),0)$ and $r'_i = (-i/(m+1),0)$ as well as blue points $b_i = (i-(m+1),(m+1)-i)$ and $b'_i = ((m+1)-i,(m+1)-i)$. We match $r_i$ to $b_i$ and $r'_i$ to $b'_i$. 
Next, we discuss important properties of an $m$-butterfly.

We call a red-on-a-line \emph{convex} matching an $n$-\emph{star} if all the $\binom{n}{2}$ pairs of segments cross. 
We say that an $n$-star \emph{looks} at a point $p$ if the blue points are all on a common line, and if $p$ is the intersection of this line with the line of the red points.
We also say that two red-blue point sets $\R,\B$ and $\R',\B'$ are \emph{fully crossing} if all the pairs of segments of the form $\{\sgt{r}{b},\sgt{r'}{b'}\}$ cross, where $(r,b,r',b') \in \R \times \B \times \R' \times \B'$.
Two matchings are fully crossing if their underlying red-blue point sets are fully crossing.
An $m$-butterfly is a red-on-a-line matching consisting of two fully crossing $m$-stars both looking at the same point $p = (0,0)$ (\Figure~\ref{fig:butterfly}(a) represents these properties but it is not drawn to scale). 

\paragraph{Intuition.}

In the following, we use the state tracking framework from Section~\ref{sec:algo} to describe how to come up with an untangle sequence starting at an $m$-butterfly with more than $\binom{2m}{2}$ flips. 
We consider a sequence of tracking choices with no $\HH \to \X$ transition (Lemma~\ref{lem:tracking}) for the long untangle sequence we build. 
We take advantage of the non-convex position of the blue points to create flip situations such as in Fig.~\ref{fig:HT}(a), where an $\HH$-pair is turned into a $\T$-pair. 

For instance, let us consider an $\X$-pair of one of the $m$-stars composing the $m$-butterfly.
At some point of the untangle sequence, we flip this $\X$-pair, turning it into an $\HH$-pair.
Later on, we turn this $\HH$-pair into a $\T$-pair, as in Fig.~\ref{fig:HT}(a).
Still later on, we turn this $\T$-pair into an $\X$-pair again, similarly to the pairs involving the horizontal segment in Fig.~\ref{fig:flip}.
This $\X$-pair will be flipped again.

We manage to carry out this whole process to flip twice all the $2 \binom{m}{2}$ pairs of the two $m$-stars composing the $m$-butterfly while still having one flip for every other pair. 
In total, we reach $\frac{3}{2}\binom{2m}{2} - \frac{m}{2}$ flips. 

\paragraph{Proof of Theorem~\ref{thm:lowerBD}.}
\label{app:lowerBD}

We prove Theorem~\ref{thm:lowerBD} with two lemmas, showing that there exists an untangle sequence of length $\frac{3}{2}\binom{2m}{2} - \frac{m}{2}$, starting at an $m$-butterfly.

\begin{lemma}[\cite{BoM16}]
  \label{lem:star}
  There exists an untangle sequence starting at any $n$-star of length $\binom{n}{2}$.
\end{lemma}
\begin{proof}
  This result has been shown in~\cite{BoM16}.
  We present a short proof for the sake of completeness.
  
  Provided we number from $1$ to $n$ the red points in their convex hull counter-clockwise order, and do the same for the blue points but clock-wise, then a red-on-a-line convex matching can be seen as a way to draw a permutation of $n$ elements.
  An inversion, then, corresponds to a crossing. 
  A bubble sort, thus, corresponds to an untangle sequence starting at such a matching. 
  
  The case of an $n$-star leads to $\binom{n}{2}$ inversion swaps, or, in other words, flips. 
\end{proof}

\begin{lemma}
  \label{lem:lowerB}
  There exists an untangle sequence starting at any $m$-butterfly of length $\frac{3}{2}\binom{2m}{2} - \frac{m}{2}$ (\Figure~\ref{fig:butterflyPart2} and \Figure~\ref{fig:butterflyPart3}).
\end{lemma}
\begin{proof}
  The untangle sequence can be divided into two phases.
  
  The first phase consists of (i) $\binom{m}{2}$ flips applied to the $m$-star submatching defined by the $m$ leftmost red points (see Lemma~\ref{lem:star}, and \Figure~\ref{fig:butterflyPart2}, steps 0 to 3), and of (ii) $\binom{m}{2}$ more flips applied to the $m$ rightmost red points (\Figure~\ref{fig:butterflyPart2}, steps 3 to 6). 
  At this point, we have two sets of $m$ crossing-free segments, each set fully crossing the other.
  
  The second phase repeats $m$ times the following routine.
  \begin{enumerate}
  \item Flip the segments defined by the innermost red points $r_1$ and $r_1'$ (\Figure~\ref{fig:butterflyPart2} and~\ref{fig:butterflyPart3}, steps 6 to 7, 11 to 12, and 16 to 17). 
  After this flip, the submatching defined by the $m$ leftmost red points and their matched points consists of $m-1$ crossing-free segments intersected by the segment from $r_1'$. A similar statement holds for the submatching defined by the $m$ rightmost red points.
  \item Untangle the submatching defined by the $m$ leftmost red points with $m-1$ flips in the following manner (\Figure~\ref{fig:butterflyPart2} and~\ref{fig:butterflyPart3}, steps 7 to 9, 12 to 14, and 17 to 19). 
  Flip of the two crossing segments with the rightmost red points, say $r_k'$ and $r_{k+1}'$ with $k \in \{1, \dots, m-1\}$, and repeat. 
  Such a flip produces a segment from $r_{k+1}'$ crossing the segments whose red points are on the left of $r_{k+1}$, and an other segment from $r_k'$ crossing none of segments of the submatching. 
  The number of crossings in the submatching decreases by $1$ unit at each flip.
  \item Similarly, untangle the $m$ rightmost red points with $m-1$ more flips (\Figure~\ref{fig:butterflyPart2} and~\ref{fig:butterflyPart3}, steps 9 to 11, 14 to 16, and 19 to 21).
  \end{enumerate}
  Each loop decreases the number of ``long'' segments (i.e., segments joining one of the leftmost red points to one of the rightmost blue points, or vice-versa) by $2$.
  At the end of the process, the left submatching is crossing-free; so is the right one; and the two of them do not intersect anymore.
  
  Summing up, the total number of flips is $2 \binom{m}{2} + 2 m (m-1) + m$. Simple calculation yields the lemma.
\end{proof}

Theorem~\ref{thm:lowerBD} follows from Lemma~\ref{lem:lowerB}.

\begin{figure}[p]
  \centering
  \includegraphics[scale=\butterflyScale,page=2]{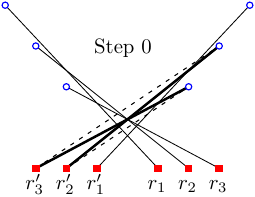}\quad
  \includegraphics[scale=\butterflyScale,page=3]{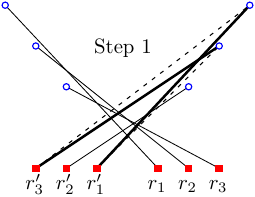}\quad
  \includegraphics[scale=\butterflyScale,page=4]{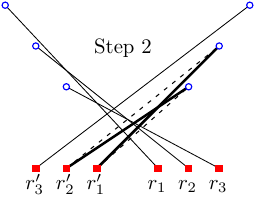}\quad
  \includegraphics[scale=\butterflyScale,page=5]{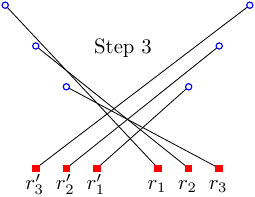}\\\vspace{\baselineskip}
  \includegraphics[scale=\butterflyScale,page=6]{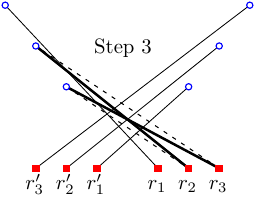}\quad
  \includegraphics[scale=\butterflyScale,page=7]{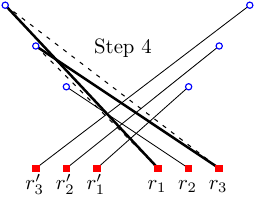}\quad
  \includegraphics[scale=\butterflyScale,page=8]{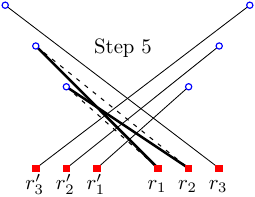}\quad
  \includegraphics[scale=\butterflyScale,page=9]{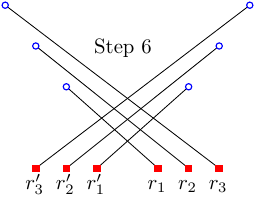}\\\vspace{\baselineskip}
  \includegraphics[scale=\butterflyScale,page=10]{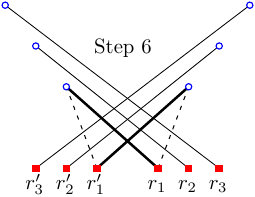}\quad
  \includegraphics[scale=\butterflyScale,page=11]{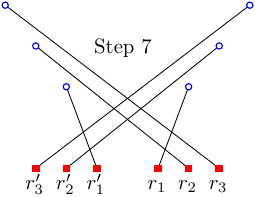}\\\vspace{\baselineskip}
  \includegraphics[scale=\butterflyScale,page=12]{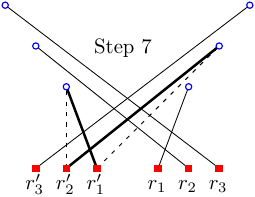}\quad
  \includegraphics[scale=\butterflyScale,page=13]{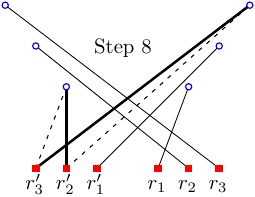}\quad
  \includegraphics[scale=\butterflyScale,page=14]{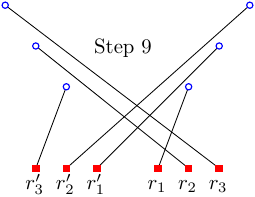}\\\vspace{\baselineskip}
  \includegraphics[scale=\butterflyScale,page=15]{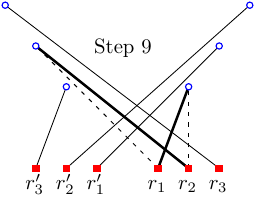}\quad
  \includegraphics[scale=\butterflyScale,page=16]{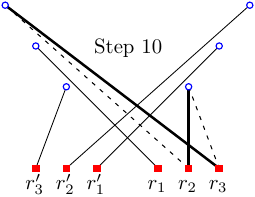}\quad
  \includegraphics[scale=\butterflyScale,page=17]{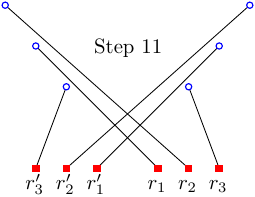}\\\vspace{\baselineskip}
  \includegraphics[scale=\butterflyScale,page=18]{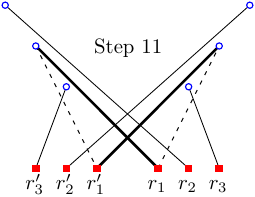}\quad
  \includegraphics[scale=\butterflyScale,page=19]{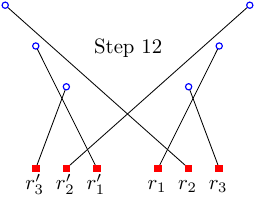}
  \caption{First part of an untangle sequence starting at a $3$-butterfly of length $21$ illustrating Lemma~\ref{lem:lowerB} and its proof. Each line corresponds to a portion of the proof, with repetitions added for clarity.} 
  \label{fig:butterflyPart2}
\end{figure}
\begin{figure}[p]
  \centering
  \includegraphics[scale=\butterflyScale,page=20]{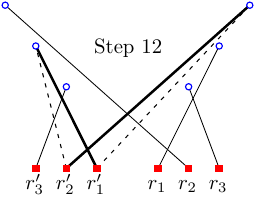}\quad
  \includegraphics[scale=\butterflyScale,page=21]{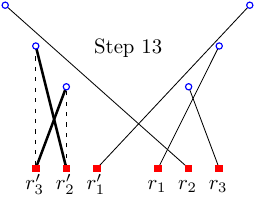}\quad
  \includegraphics[scale=\butterflyScale,page=22]{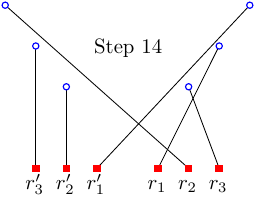}\\\vspace{\baselineskip}
  \includegraphics[scale=\butterflyScale,page=23]{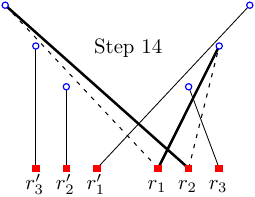}\quad
  \includegraphics[scale=\butterflyScale,page=24]{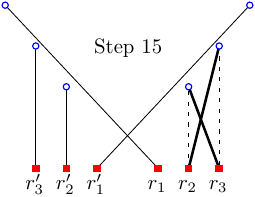}\quad
  \includegraphics[scale=\butterflyScale,page=25]{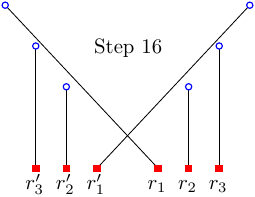}\\\vspace{\baselineskip}
  \includegraphics[scale=\butterflyScale,page=26]{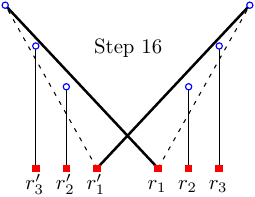}\quad
  \includegraphics[scale=\butterflyScale,page=27]{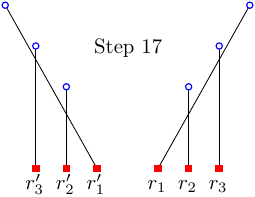}\\\vspace{\baselineskip}
  \includegraphics[scale=\butterflyScale,page=28]{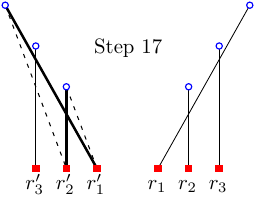}\quad
  \includegraphics[scale=\butterflyScale,page=29]{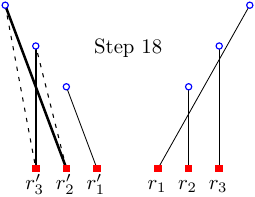}\quad
  \includegraphics[scale=\butterflyScale,page=30]{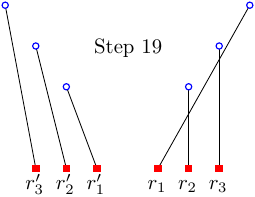}\\\vspace{\baselineskip}
  \includegraphics[scale=\butterflyScale,page=31]{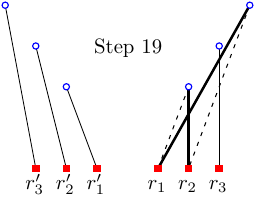}\quad
  \includegraphics[scale=\butterflyScale,page=32]{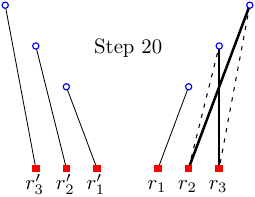}\quad
  \includegraphics[scale=\butterflyScale,page=33]{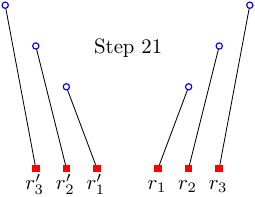}
  \caption{Second part of an untangle sequence starting at a $3$-butterfly of length $21$ illustrating Lemma~\ref{lem:lowerB} and its proof. Each line corresponds to a portion of the proof, with repetitions added for clarity.} 
  \label{fig:butterflyPart3}
\end{figure}

\subsection{Lower Bound on \texorpdfstring{$\dd(n)$}{d(n)}}
\label{sec:fenceLowerB}

We provide a convex red-blue matching which we call an \emph{$m$-fence}, with $2m$ segments and $3 m - 2$ crossings (\Figure~\ref{fig:fence}(b)). 
Next, we give the precise definition of an $m$-fence, together with some useful terminology.  
Then, we prove Theorem~\ref{thm:lowerBd} with three lemmas inferring that all untangle sequences starting at an $m$-fence have length $3 m - 2$, that is, each flip reduces the number of crossings by exactly one unit. 

\paragraph{Fence.}

Let $q_{2m+2}$, $q_{2m}$, $q_{2m-1}, \dots$, $q_4$, $q_3$, $q_1$, $p_1$, $p_3$, $p_4, \dots$, $p_{2m-1}$, $p_{2m}$, $p_{2m+2}$ be $4m$ points in convex position, ordered counter-clockwise, and with colors alternating every two points (\Figure~\ref{fig:fence}(b)). 
More precisely, points $p_i,q_i$ are red if $i \equiv 1,2 \mod 4$ and blue otherwise.
We deliberately avoid using the indices $2$ and $2m+1$ to simplify the description.
The segments of an $m$-fence are the $\sgt{p_i}{q_{i+3}}$ and the $\sgt{q_i}{p_{i+3}}$ where $i$ is odd and varies between $1$ and $2m-1$.

For $1 \leq k \leq m+1$, the \emph{$k$-th column} consists of the at most $4$ points with indices $2k-1$ and $2k$.
We say that a convex red-blue matching with the same point set as an $m$-fence is a \emph{derived $m$-fence} if,
for all $k \in \{2, \dots , m\}$, for all $w \in \{p, q\}$, one of the following statements holds:
\begin{enumerate}
    \item \label{item:v} $w_{2k-1}$ is matched to a point of the $(k-1)$-th column, and $w_{2k}$ is matched to a point of the $(k+1)$-th column, or
    \item \label{item:x} $w_{2k-1}$ is matched to a point of the $(k+1)$-th column, and $w_{2k}$ is matched to a point of the $(k-1)$-th column.
\end{enumerate}

Five examples of derived $m$-fences are presented in \Figure~\ref{fig:derivedFence}.
Note that an $m$-fence is in particular a derived $m$-fence. 

When statement~\ref{item:x} holds, the two segments cross.
We call such a crossing an \emph{end} crossing.
Similarly, a \emph{middle} crossing is a crossing of the form $\{\sgt{p_i}{q_j},\ \sgt{q_{i'}}{p_{j'}}\}$, where $i$ and $i'$ are of the same column, and $j$ and $j'$ are of the same column.

\paragraph{Proof of Theorem~\ref{thm:lowerBd}.}

To prove Theorem~\ref{thm:lowerBd}, we first show with two lemmas that a flip changes a derived $m$-fence into another derived $m$-fence. Finally, we show that a flip of a derived $m$-fence reduces its number of crossings by exactly one unit.

\begin{figure}[!ht]
  \centering
  \includegraphics[page=5,scale=\graphicsScale]{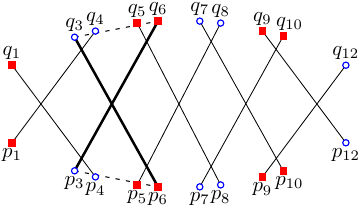}\\\vspace{\baselineskip}
  \includegraphics[page=6,scale=\graphicsScale]{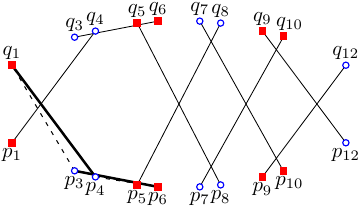}\qquad
  \includegraphics[page=7,scale=\graphicsScale]{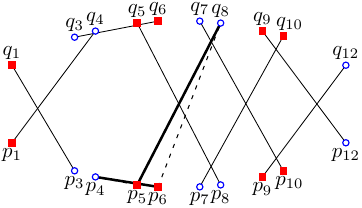}\\ \vspace{\baselineskip}
  \includegraphics[page=8,scale=\graphicsScale]{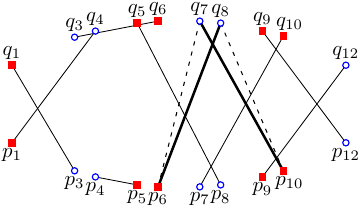}\qquad
  \includegraphics[page=9,scale=\graphicsScale]{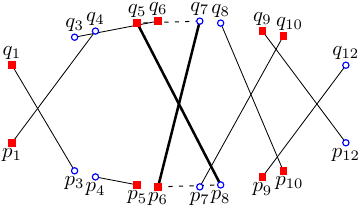}\\\vspace{\baselineskip}
  \includegraphics[page=10,scale=\graphicsScale]{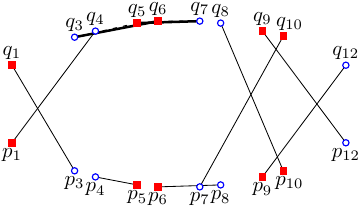}\qquad
  \includegraphics[page=11,scale=\graphicsScale]{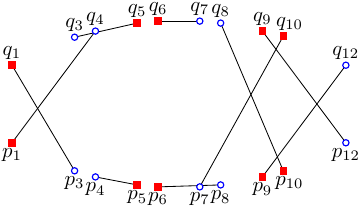}
  \caption{The beginning of an untangle sequence starting at a $5$-fence. It is composed of derived $5$-fences.}
  \label{fig:derivedFence}
\end{figure}

\begin{lemma}
   \label{lem:crossingsDerivedFence}
   A crossing in a derived $m$-fence is either an end crossing or a middle crossing.
\end{lemma}
\begin{proof}
  The definition of a derived $m$-fence implies that a crossing must involve two or three consecutive columns. 
  If exactly three columns are involved, the same definition excludes any crossing aside from the end crossings.
  If exactly two columns are involved, the definition again excludes any crossing aside from the middle crossings.
\end{proof}

\begin{lemma}
   \label{lem:flipStableDerivedFence}
   A flip changes a derived $m$-fence into another derived $m$-fence.
\end{lemma}
\begin{proof}
  Lemma~\ref{lem:crossingsDerivedFence} ensures that we only have the following two cases.
  (i) The flip of an end crossing on the $w_0$ side ($w_0 \in \{p, q\}$) of the $k_0$-th column only changes statement~\ref{item:x} of the definition of a derived $m$-fence into statement~\ref{item:v} for $k,r=k_0,r_0$.
  The statements for the other $k,r$ are unchanged.
  (ii) The flip of a middle crossing simply leaves unchanged the statements for all $k,r$.
  
  \Figure~\ref{fig:derivedFence} is actually a sequence of flips starting at an $m$-fence and it contains essentially all the possible cases (symmetries aside). 
\end{proof}

\begin{lemma}
   \label{lem:flipDrops1CrossingInDerivedFence}
   A flip of a derived $m$-fence reduces its number of crossings by exactly one unit.
\end{lemma}
\begin{proof}
  Let $\M$ be a derived $m$-fence.
  Let $s_1$ and $s_2$ be two crossing segments of $\M$.
  Let $s$ be any other segment of $\M$.
  Let $s_1'$ and $s_2'$ be the two segments replacing $s_1$ and $s_2$ after they have been flipped, changing $\M$ into $\M'$.
  We show that the number of crossings between $s$ and $\pair{s_1}{s_2}$ is the same as between $s$ and $\pair{s_1'}{s_2'}$, ensuring that $\M'$ has exactly $1$ crossing less than $\M$.
  
  Let us recall that, as for any convex matching, the number of crossings cannot increase~\cite{BMS19}. 
  The proof of this result consists of the analysis of the five possible typical convex matchings (symmetries aside) of the three segments $s_1, s_2, s$ (\Figure~\ref{fig:5ConvexConfigurations}).
  It is notable that only one of these five matchings, the one where each endpoint of $s$ lies in between two endpoints of $\pair{s_1}{s_2}$ of the same color, corresponds to an actual decrease in the number of crossings involving $s$.

  \begin{figure}[!ht]
  \centering
  \includegraphics[page=1,scale=\graphicsScale]{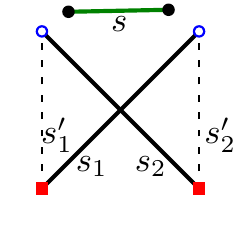} \quad 
  \includegraphics[page=2,scale=\graphicsScale]{5ConvexConfigurations} \quad 
  \includegraphics[page=3,scale=\graphicsScale]{5ConvexConfigurations} \quad 
  \includegraphics[page=4,scale=\graphicsScale]{5ConvexConfigurations} \quad 
  \includegraphics[page=5,scale=\graphicsScale]{5ConvexConfigurations}
  \caption{The five convex positions of $s$ with respect to the flipping pair $\pair{s_1}{s_2}$. This is used in the proof of Lemma~\ref{lem:flipDrops1CrossingInDerivedFence}.}
  \label{fig:5ConvexConfigurations}
  \end{figure}

  This crossing-destructive case cannot occur if two endpoints of $\pair{s_1}{s_2}$ of the same color are adjacent on the convex hull.
  Thus, Lemma~\ref{lem:flipDrops1CrossingInDerivedFence} holds for flips of end crossings.
  
  If $\pair{s_1}{s_2}$ is a middle crossing, then, by definition of a derived $m$-fence, no segment $s$ intersects both $s_1$ and $s_2$.
  Thus, the crossing-destructive case cannot occur, and Lemma~\ref{lem:flipDrops1CrossingInDerivedFence} holds for flips of middle crossings.
\end{proof}

Theorem~\ref{thm:lowerBd} follows from Lemma~\ref{lem:flipStableDerivedFence} and Lemma~\ref{lem:flipDrops1CrossingInDerivedFence}.

\section{Concluding Remarks}
\label{sec:conclusion}

Untangle sequences of TSP tours have been investigated since the 80s, when a cubic upper bound on $\D(n)$ has been discovered~\cite{VLe81}. This bound also holds for matchings (even non-bipartite ones) and has not been improved ever since.
Except for the convex case, there are big gaps between the lower and upper bounds, as can be seen in Table~\ref{tab:summary}. Experiments on tours and matchings have shown that, in all cases tested, the cubic upper bound is not tight and the lower bounds seem to be asymptotically tight.

Untangle sequences have many unexpected properties which make the problem harder than it seems at first sight. 
The following questions remain open.
\begin{enumerate}
    \item If we add a new segment to a crossing-free matching, what is the maximum length of an untangle sequence? Notice that an $\oo(n^2)$ bound would lead to an $\oo(n^3)$ bound for $\dd(n)$.
    
    \item Is it always possible to find an untangle sequence that does not flip the same pair of segments twice? Using a balancing argument, we can show that the number of \emph{distinct} flips in any untangle sequence is $\OO(n^{8/3})$~\cite{FGR22}.
    
    \item What is the maximum number of flips involving a given point? The cubic potential provides a quadratic bound which leads again to an $\OO(n^3)$ bound for $\D(n)$.
\end{enumerate}

We proved the NP-hardness of computing the shortest untangle sequence for a red-blue matching. 
What is the complexity of computing the shortest untangle sequence for a TSP tour, for a red-on-a-line matching, or even for a convex instance? 
What about the longest untangle sequence?

\bibliography{ref}

\begin{thebibliography}{10}

\bibitem{AMP15}
Oswin Aichholzer, Wolfgang Mulzer, and Alexander Pilz.
\newblock Flip distance between triangulations of a simple polygon is
  {NP}-complete.
\newblock {\em Discrete \& Computational Geometry}, 54(2):368--389, 2015.

\bibitem{AkAl89}
Jin Akiyama and Noga Alon.
\newblock Disjoint simplices and geometric hypergraphs.
\newblock In {\em Third international conference on Combinatorial mathematics},
  pages 1--3, 1989.

\bibitem{BeI08}
Sergey Bereg and Hiro Ito.
\newblock Transforming graphs with the same degree sequence.
\newblock In {\em Computational Geometry and Graph Theory}, pages 25--32, 2008.

\bibitem{BeI17}
Sergey Bereg and Hiro Ito.
\newblock Transforming graphs with the same graphic sequence.
\newblock {\em Journal of Information Processing}, 25:627--633, 2017.

\bibitem{BMS19}
Ahmad Biniaz, Anil Maheshwari, and Michiel Smid.
\newblock Flip distance to some plane configurations.
\newblock {\em Computational Geometry}, 81:12--21, 2019.

\bibitem{BBH19}
Marthe Bonamy, Nicolas Bousquet, Marc Heinrich, Takehiro Ito, Yusuke Kobayashi,
  Arnaud Mary, Moritz M{\"{u}}hlenthaler, and Kunihiro Wasa.
\newblock The perfect matching reconfiguration problem.
\newblock In {\em 44th International Symposium on Mathematical Foundations of
  Computer Science}, volume 138 of {\em LIPIcs}, pages 80:1--80:14, 2019.

\bibitem{BoM16}
{\'{E}}douard Bonnet and Tillmann Miltzow.
\newblock Flip distance to a non-crossing perfect matching.
\newblock {\em Computing Research Repository}, abs/1601.05989, 2016.

\bibitem{BOSE200960}
Prosenjit Bose and Ferran Hurtado.
\newblock Flips in planar graphs.
\newblock {\em Computational Geometry}, 42(1):60--80, 2009.

\bibitem{BJ20}
Nicolas Bousquet and Alice Joffard.
\newblock Approximating shortest connected graph transformation for trees.
\newblock In {\em Theory and Practice of Computer Science}, pages 76--87, 2020.

\bibitem{FGR22}
Guilherme~D. da~Fonseca, Yan Gerard, and Bastien Rivier.
\newblock On the longest flip sequence to untangle segments in the plane, 2022.
\newblock URL: \url{https://arxiv.org/abs/2210.12036}, \href
  {https://doi.org/10.48550/ARXIV.2210.12036}
  {\path{doi:10.48550/ARXIV.2210.12036}}.

\bibitem{deBerg2012}
Mark De~Berg and Amirali Khosravi.
\newblock Optimal binary space partitions for segments in the plane.
\newblock {\em International Journal of Computational Geometry \&
  Applications}, 22(03):187--205, 2012.

\bibitem{ERV14}
Matthias Englert, Heiko R{\"o}glin, and Berthold V{\"o}cking.
\newblock Worst case and probabilistic analysis of the {2-Opt} algorithm for
  the {TSP}.
\newblock {\em Algorithmica}, 68(1):190--264, 2014.

\bibitem{EKM13}
P{\'e}ter~L Erd{\H{o}}s, Zolt{\'a}n Kir{\'a}ly, and Istv{\'a}n Mikl{\'o}s.
\newblock On the swap-distances of different realizations of a graphical degree
  sequence.
\newblock {\em Combinatorics, Probability and Computing}, 22(3):366--383, 2013.

\bibitem{Hak62}
Seifollah~Louis Hakimi.
\newblock On realizability of a set of integers as degrees of the vertices of a
  linear graph. {I}.
\newblock {\em Journal of the Society for Industrial and Applied Mathematics},
  10(3):496--506, 1962.

\bibitem{Hak63}
Seifollah~Louis Hakimi.
\newblock On realizability of a set of integers as degrees of the vertices of a
  linear graph {II}. uniqueness.
\newblock {\em Journal of the Society for Industrial and Applied Mathematics},
  11(1):135--147, 1963.

\bibitem{HeSu92}
John Hershberger and Subhash Suri.
\newblock Applications of a semi-dynamic convex hull algorithm.
\newblock {\em BIT Numerical Mathematics}, 32(2):249--267, 1992.

\bibitem{HNU99}
Ferran Hurtado, Marc Noy, and Jorge Urrutia.
\newblock Flipping edges in triangulations.
\newblock {\em Discrete \& Computational Geometry}, 22(3):333--346, 1999.

\bibitem{TDH11}
Takehiro Ito, Erik~D. Demaine, Nicholas~J.A. Harvey, Christos~H. Papadimitriou,
  Martha Sideri, Ryuhei Uehara, and Yushi Uno.
\newblock On the complexity of reconfiguration problems.
\newblock {\em Theoretical Computer Science}, 412(12):1054--1065, 2011.

\bibitem{phdJof}
Alice Joffard.
\newblock {\em Graph domination and reconfiguration problems}.
\newblock PhD thesis, Université Claude Bernard Lyon 1, 2020.

\bibitem{Law72}
Charles~L Lawson.
\newblock Transforming triangulations.
\newblock {\em Discrete Mathematics}, 3(4):365--372, 1972.

\bibitem{LuP15}
Anna Lubiw and Vinayak Pathak.
\newblock Flip distance between two triangulations of a point set is
  {NP}-complete.
\newblock {\em Computational Geometry}, 49:17--23, 2015.

\bibitem{NiN18}
Naomi Nishimura.
\newblock Introduction to reconfiguration.
\newblock {\em Algorithms}, 11(4), 2018.

\bibitem{OdW07}
Yoshiaki Oda and Mamoru Watanabe.
\newblock The number of flips required to obtain non-crossing convex cycles.
\newblock In {\em Kyoto International Conference on Computational Geometry and
  Graph Theory}, pages 155--165, 2007.

\bibitem{Pil14}
Alexander Pilz.
\newblock Flip distance between triangulations of a planar point set is
  apx-hard.
\newblock {\em Computational Geometry}, 47(5):589--604, 2014.

\bibitem{Heu13}
Jan van~den Heuvel.
\newblock The complexity of change.
\newblock {\em Surveys in Combinatorics}, 409:127--160, 2013.

\bibitem{VLe81}
Jan van Leeuwen.
\newblock Untangling a traveling salesman tour in the plane.
\newblock In {\em 7th Workshop on Graph-Theoretic Concepts in Computer
  Science}, 1981.

\bibitem{Wil99}
Todd~G Will.
\newblock Switching distance between graphs with the same degrees.
\newblock {\em SIAM Journal on Discrete Mathematics}, 12(3):298--306, 1999.

\bibitem{WCL09}
Ro{-}Yu Wu, Jou{-}Ming Chang, and Jia{-}Huei Lin.
\newblock On the maximum switching number to obtain non-crossing convex cycles.
\newblock In {\em 26th Workshop on Combinatorial Mathematics and Computation
  Theory}, pages 266--273, 2009.

\end{thebibliography}

\end{document}